\documentclass[1p,times]{elsarticle}

\usepackage[utf8]{inputenc}
\usepackage{enumitem}
\usepackage{verbatim}
\usepackage{subcaption}
\usepackage[usenames,svgnames,dvipsnames]{xcolor} 
\usepackage{tcolorbox}
\definecolor{labelkey}{rgb}{0.5,0.5,0.5}
\usepackage{comment}
\usepackage{mathtools}
\usepackage{breqn}
\usepackage{amsmath,amssymb,amstext,amsthm}
\usepackage{amsbsy,bm}
\usepackage{amsfonts}

\usepackage{tikz}
\usepackage{bbm}
\usepackage{hyperref}
\hypersetup{
    colorlinks=true, 
    linktoc=all,    
    linkcolor=blue,  
}
\theoremstyle{plain}
\newtheorem{theorem}{Theorem}[section]
\newtheorem{lem}{Lemma}

\newtheorem{lemma}[theorem]{Lemma}
\newtheorem{corollary}[theorem]{Corollary}

\theoremstyle{definition}

\theoremstyle{remark}
\newtheorem*{remark2}{Remark}
\newtheorem*{example}{Example}

\newcommand{\pr}[1]{\mathbb{P}\left( #1 \right)}
\newcommand{\mean}[1]{\mathbb{E}\left[{#1}\right]}

\newcommand{\ceil}[1]{\left \lceil #1 \right \rceil}

\newcommand{\Order}{{\rm O}}
\newcommand{\order}{{\rm o}}

\renewcommand{\cal}{\mathcal}

\newcommand{\eps}{\varepsilon}

\newcommand{\bbox}{\begin{tcolorbox}}
\newcommand{\ebox}{\end{tcolorbox}}

\newenvironment{myitemize}
{\begin{itemize}
\setlength{\itemsep}{0pt}
\setlength{\parskip}{0pt}
\setlength{\parsep}{0pt} 	}
{ \end{itemize} }

\newenvironment{enum}
{\begin{enumerate}
\setlength{\itemsep}{0pt}
\setlength{\parskip}{0pt}
\setlength{\parsep}{0pt} 	}
{ \end{enumerate} }

\newcommand{\bit}{\begin{myitemize}}
\newcommand{\eit}{\end{myitemize}}
\newcommand{\ben}{\begin{enum}}
\newcommand{\een}{\end{enum}}

\newcommand{\beq}[1]{\begin{equation}\label{#1}}
\newcommand{\eeq}{\end{equation}}

\newcommand{\adjacent}{\textsf{adjacent}}
\newcommand{\degree}{\textsf{degree}}
\newcommand{\neighbor}{\textsf{neighborhood}}

\newcommand{\rank}{\textsf{rank}}
\newcommand{\select}{\textsf{select}}
\newcommand{\access}{\textsf{access}}

\newcommand{\cnt}{\textsf{count}}
\newcommand{\report}{\textsf{report}}

\newcommand{\rmq}{\textsf{rmq}}
\newcommand{\rMq}{\textsf{rMq}}

\newcommand {\BL} {\begin{lem}} 
\newcommand {\EL} {\end{lem}} 

\newtheorem{cor}  {Corollary} 
\newcommand {\BCR} {\begin{cor}}
\newcommand {\ECR} {\end{cor}}

\newtheorem{thm} {Theorem} 
\newcommand {\BT} {\begin{thm}}
\newcommand {\ET} {\end{thm}}

\newtheorem{defi} {Problem} 
\newcommand {\BDE} {\begin{defi}}
\newcommand {\EDE} {\end{defi}}

\newtheorem{xxx} {Definition} 
\newcommand {\BD} {\begin{xxx}}
\newcommand {\ED} {\end{xxx}}

\newtheorem{xxxxx} {Observation} 
\newcommand {\BO} {\begin{xxxxx}}
\newcommand {\EO} {\end{xxxxx}}

\journal{Elsevier}

\begin{document}

\begin{frontmatter}

\title{
Succinct Navigational Oracles for Families  of Intersection Graphs on a Circle} 
\author[a]{H{\"{u}}seyin Acan}
\ead{huseyin.acan@drexel.edu}
\author[b]{Sankardeep Chakraborty}
\ead{sankardeep.chakraborty@gmail.com}
\author[c]{Seungbum Jo}
\ead{sbjo@chungbuk.ac.kr}
\author[d]{Kei Nakashima}
\ead{kei\_nakashima@mist.i.u-tokyo.ac.jp}
\author[d]{Kunihiko Sadakane}
\ead{sada@mist.i.u-tokyo.ac.jp}
\author[e]{Srinivasa Rao Satti}
\ead{ssrao@tcs.snu.ac.kr}
\address[a]{Drexel University, United States}
\address[b]{National Institute of Informatics, Japan}
\address[c]{Chungbuk National University, South Korea}
\address[d]{The University of Tokyo, Japan}
\address[e]{Seoul National University, Seoul, South Korea}

\sloppy
\begin{abstract}
We consider the problem of designing succinct navigational oracles, i.e., succinct data structures supporting basic navigational queries such as degree, adjacency and neighborhood efficiently for intersection graphs on  a circle, which include graph classes such as {\it circle graphs}, {\it $k$-polygon-circle graphs}, {\it circle-trapezoid graphs}, {\it trapezoid graphs}.
The degree query reports the number of incident edges to a given vertex, the adjacency query asks if there is an edge between two given vertices, and the neighborhood query enumerates all the neighbors of a given vertex. 
We first prove a general lower bound for  these intersection graph classes, and then present a uniform approach that lets us obtain matching lower and upper bounds for representing each of these graph classes. More specifically, our lower bound proofs use a unified technique to produce tight bounds for all these classes, and this is followed by our data structures which are also obtained from a unified representation method to achieve succinctness for each class. 
In addition, we prove a lower bound of space for representing {\it trapezoid} graphs, and give a succinct navigational oracle for this class of graphs.
\end{abstract}
\end{frontmatter}

\section{Introduction}
Intersection graphs of geometric objects are fascinating combinatorial objects from the point of view of algorithmic graph theory as many hard ({\sf NP-complete} in general) optimization problems become easy, i.e., polynomially solvable when restricted to various classes of intersection graphs. Thus, they provide us with clues with respect to the line of demarcation between {\sf P} and {\sf NP}, if there exists such a line. Furthermore, they also have a broad range of practical applications~\cite[Chapter 16]{spinrad}. 
Perhaps the simplest and most widely studied such objects are the interval graphs, intersection graphs of intervals on a line~\cite{Hajos,Golumbic85,Golumbic}. Several characterizations of interval graphs~\cite{Golumbic} including their linear time recognition algorithms are already known in the literature~\cite{HMPV}. There exist many generalizations of interval graphs, and we focus particularly in this work on some of these generalizations involving intersection of geometric objects bound to a circle. 

More specifically, we study circle graphs, $k$-polygon-circle graphs, circle-trapezoid graphs, 
and trapezoid graphs
in this article. A circle graph is defined as the intersection graph of chords in a circle~\cite{bouchet,even}.
{\it Polygon-circle} graphs~\cite{koebe} are the intersection graphs of convex polygons inscribed into a circle, and the special case, when all the convex polygons have exactly $k$ corners, we call the intersection graph $k$-polygon-circle~\cite{EnrightK19}. 
Circle-trapezoid graphs are the intersection graphs of circle trapezoids on a common circle, 
where a circle trapezoid is defined as the convex hull of two disjoint arcs on the circle~\cite{FelsnerMW97}. 
Finally, trapezoid graphs are the intersection graphs of trapezoids between two parallel lines
which can be regarded as a circle with a sufficiently large radius.
These graphs are not only theoretically interesting to study but they also show up in important practical 
application domains, e.g., in VLSI physical layout~\cite{spinrad,Golumbic}. In spite of having such importance and being such basic geometric graphs, we are not aware of any study of these aforementioned objects using the lens of \textit{succinct data structures}~\cite{gonzalo} where we need to achieve the following twofold tasks. The first goal is to bound from below the cardinality of a set $T$ consisting of combinatorial objects with certain property, and this is followed by storing any arbitrary member $x \in T$ using the information theoretic minimum of $\log(|T|)+o(\log(|T|))$ bits (throughout this paper, $\log$ denotes the logarithm to the base $2$) while still being able to support the relevant set of queries efficiently on $x$, both the tasks we focus on here. We assume the usual model of computation, namely a $\Theta (\log n)$-bit word RAM model where $n$ is the size of the input. This is a standard assumption that implies a vertex can be distinguished, in constant time, with a label that fits within a word of the RAM. Finally all the graphs we deal with in this paper are simple, undirected, unlabeled and unweighted.

\subsection{Related Work}
\noindent
{\bf Succinct navigational oracles.} There already exists a huge body of work on representing several classes of graphs succinctly along with supporting basic navigational queries efficiently. A partial list of such special graph classes would be arbitrary graphs~\cite{FarzanM13}, trees~\cite{MunroR01}, planar graphs~\cite{AleardiDS08}, chordal graphs~\cite{MunroW18}, 
graphs with bounded tree-width $k$ (partial $k$-trees)~\cite{FarzanK14}, etc.
Specially, one can consider (i) circular-arc graphs (intersection graphs on the arcs on a circle), (ii) interval graphs (a sub-class of circular-arc graphs), and (iii) permutation graphs (intersection graphs of line segments between two parallel lines) as the special case of the intersection graphs on a circle. 
For interval graphs and circular-arc graphs, Gavoille and Paul~\cite{Gavoille} (and independently, Acan et  al.~\cite{DBLP:conf/wads/AcanCJS19}) showed that $n \log n - \Order(n \log \log n)$ bits are necessary for representing an interval or a circular-arc graph with $n$ vertices. In~\cite{DBLP:conf/wads/AcanCJS19}, the authors also presented succinct navigation oracles for both graph classes. Also for permutation graphs, a lower bound of $(n \log n - \Order(n \log \log))$ bits is known~\cite{BazzaroG09,KohR07}.

\noindent
{\bf Algorithmic graph-theoretic results.} All the intersection graphs that we focus in this paper are very well studied in the algorithmic graph theory literature. Circle graphs (which are essentially same as {\it overlap graphs}\footnote{https://www.graphclasses.org/classes/gc_913.html}) can be recognized in polynomial time along with admitting polynomial time algorithms for various optimization problems like feedback vertex set and independent set (see~\cite{spinrad} and references therein for more details). These graphs were first introduced in the early 1970s, under the name {\it alternance graphs}, as a tool used for sorting permutations using stacks~\cite{even}. The introduction of polygon-circle graphs (which are same as {\it spider graphs}~\cite{koebe}) was motivated by the fact that this class of graphs is closed under taking induced minors. Even though the problem of recognising polygon-circle and $k$-polygon-circle graphs is {\sf NP-complete}~\cite{Pergel07, DBLP:conf/gd/KratochvilP03}, many optimization problems that are otherwise {\sf NP-Complete} on general graphs can be solved in polynomial time given a polygon-circle representation of a graph (see~\cite{spinrad} for more details). Felsner et al.~\cite{FelsnerMW97} introduced circle-trapezoid graphs as an extension of trapezoid graphs and devised polynomial time algorithms for maximum weighted clique and maximum weighted independent set problems.  We refer the reader to~\cite{Golumbic85,Golumbic,spinrad,terry} for more details on these graph classes and other related problems.

\subsection{Our Results}

In this paper, we consider a graph class defined as the intersection graphs of objects on a circle,
where objects are \emph{generalized polygons},
polygons whose corners are on the circle
and edges are either chords or arcs of the circle.
This contains many graph classes including
(1) circular-arc graphs,
(2)~\textit{$k$-polygon-circle graphs}, which are intersection graphs of polygons on a circle, where every polygon has $k$ chords, and
(3) circle graphs
(4) circle-trapezoid graphs.

Note that these example classes correspond to $k$-polygon circle graphs with a fixed $k$, while our upper and lower bounds in fact apply to a more general case when the graph contains polygons with different number of corners.

We first show a space lower bound for representing
the above graph classes (Theorem~\ref{th:generalLB}).
For circle graphs,
we show that the lower bound from Section~\ref{sec:generalLB} can be improved to $n \log n - \Order(n)$ bits.
Furthermore using a similar idea to prove Theorem~\ref{th:generalLB}, we also obtain a space lower bound for representing trapezoid graphs.
These lower bound results 
are summarized in Table~\ref{table:LB}. 

\begin{table}[tb]
	\centering
	\caption{Lower bounds of families of intersection graphs.}
	\label{table:LB}
	\begin{tabular}{|c|c|c|}
		\hline
		Graph class & Space lower bound (in bits) & Reference (this paper) \\ \hline
		circle & $n \log n - \Order(n)$ & Theorem~\ref{thm:lb_circle} \\
	    $k$-polygon-circle & $(k-1)n \log n - \Order(kn \log\log n)$ & Theorem~\ref{th:generalLB}, $k={\rm polylog}(n)$\\
		circle-trapezoid & $3 n \log n - 4 \log\log n - \Order(n)$ & Corollary~\ref{cor:1} \\ 
		trapezoid & $3 n \log n - 4 \log\log n - \Order(n)$ & Lemma~\ref{lem:trapezoid_lb}\\
        \hline
	\end{tabular}
\end{table}

Next, we consider data structures for representing families of intersection graphs on a circle
which support three basic navigation queries efficiently, which are defined as follows. Given a graph $G=(V,E)$ such that $|V|=n$ and two vertices $u,v\in V$, (i) $\degree{}(v)$ query returns the number of vertices that are adjacent to $v$ in $G$, (ii) $\adjacent{}(u,v)$ query returns true if $u$ and $v$ are adjacent in $G$, and false otherwise, and finally (iii) $\neighbor{}(v)$ query returns all the vertices that are adjacent to $v$ in $G$. 

We give a unified representation of families of intersection graphs of generalized polygons on a circle where generalized polygon is define as a shape where every pair of consecutive corners are connected by either an arc or a chord on a circle.
From these results,
we can obtain succinct data structures which can support $\adjacent{}$, $\degree{}$, and $\neighbor{}$ queries efficiently, for all the graphs classes in Table~\ref{table:LB}, including interval and permutation graphs. Note that for the graph classes in Table~\ref{table:LB}, these are the first succinct data structures.

Finally, for circle graphs and trapezoid graphs, we present alternative succinct data structures which support faster $\degree{}$ queries (for  vertices whose degree is $\Omega(\log {n} / \log \log {n})$).



\subsection{Paper Organization}
After listing preliminary data structures
that will be used throughout our paper in Section~\ref{prelims}, we move on to present the central contributions of our work. 
In Section~\ref{sec:general}, we prove the lower bound of space to represent intersection graphs of generalized polygons on a circle, from which the lower bound results in Table~\ref{table:LB} for $k$-polygon-circle and circle-trapezoid graphs follow,
and present our general upper bound result (see Theorem~\ref{thm:generaluppernew}) that provides succinct data structures for all these graphs in a unified manner. 
In Section~\ref{sec:circle}, we give a space lower bound for representing circle graphs which improves the lower bound obtained from Theorem~\ref{th:generalLB}, and also give an alternative succinct representation for circle graphs. 
In Section~\ref{sec:trapezoid}, we give a space lower bound for representing trapezoid graphs, and augment it with an alternative succinct representation for trapezoid graphs. 
Finally, we conclude in Section~\ref{conclusion} with some open problems.
\section{Preliminaries}\label{prelims}
In this section, we introduce some data structures that will be used in the rest of the paper.
\newline
\noindent\textbf{Rank, Select and Access queries.}
Let $A[1 \dots, n]$ be an array of size $n$ over an alphabet $\Sigma = \{0, 1, \dots, \sigma-1\}$ of size $\sigma$. 
Then for $ 1 \le i \le n$ and $\alpha \in \Sigma$, we define the \rank{}, \select{} and \access{} queries on $A$ as follows. 
\begin{itemize}
    \item $\rank_{\alpha}(i, A)$ returns the number of occurrences of $\alpha$ in $A[1  \dots i]$.
	\item $\select_{\alpha}(i, A)$ returns the position $j$ where $A[j]$ is the $i$-th $\alpha$ in $A$.
	\item $\access(i, A)$ returns $A[i]$.
\end{itemize}


Then, the following data structures are known for supporting the above queries.

\begin{lem}[\cite{Clark:1996:EST:313852.314087}]
\label{rankselect2}
Given a bit array $B[1 \dots n]$ of size $n$, there exists an $n + o(n)$-bit data structure which answers $\rank{}_{\alpha}$, $\select{}_{\alpha}$ for $\alpha=\{0, 1\}$, and $\access{}$ queries on $B$ in $O(1)$ time.
\end{lem}


\begin{lem}[\cite{BCGNNalgor13}]
\label{rankselect}
Given an array $A[1 \dots n]$ over
$\Sigma = \{0, 1, \dots, \sigma-1\}$ for any $\sigma > 1$, there exists an $n H_0 + 
o(n) \cdot O(H_0+1)$-bit 
data structure that answers $\rank{}_{\alpha}$ and $\access{}$ queries in $O(1+ \log \log \sigma)$ time and $\select{}_{\alpha}$ queries in $O(1)$ time on $S$,
for any $\alpha \in \Sigma$, 
where $H_0 \le \log \sigma$ is the order-$0$ entropy of $A$.
\end{lem}




\noindent\textbf{Range minimum and maximum queries.}
Let $A[1 \dots, n]$ be an array of size $n$ over a totally ordered set.
Then for $ 1 \le i \le j \le n$, we define the \rmq{}, \rMq{} queries on $A$ as follows. 
\begin{itemize}
    \item $\rmq(A, i, j)$: returns the index $m$ of $A$ that attains the minimum value $A[m]$
    in $A[i \dots j]$.  If there is a tie, returns the leftmost one.
    \item $\rMq(A, i, j)$: returns the index $m$ of $A$ that attains the maximum value $A[m]$
    in $A[i \dots j]$.  If there is a tie, returns the leftmost one.
\end{itemize}

\begin{lem}[\cite{DBLP:journals/siamcomp/FischerH11}]
\label{rmq}
Given an array $A[1 \dots n]$ of size $n$ over a totally ordered set, there exists a
$2n + o(n)$-bit data structure which answers $\rmq(A, i, j)$ queries
in $O(1)$ time.
\end{lem}

Note that the above structure does not access $A$ at query time.
Similarly, one can also obtain a $2n + o(n)$-bit data structure supporting range maximum queries in $O(1)$ time.

\section{Unified Lower and Upper Bounds}
\label{sec:general}
In this section, we give a unified representation of families of
intersection graphs of generalized polygons on a circle.
Here, we define a generalized polygon as a shape where every pair of consecutive corners are connected by either an arc or a chord on the circle. 
We assume that two arcs not adjacent, otherwise we can merge them into a single arc.
Note that we define a single chord (or an arc) as a polygon with two corners.
Since there is no restriction on the number of corners for each polygon, this graph is a generalization of circle, $k$-polygon-circle and circle-trapezoid graphs.
We note that a circular-arc graph can be represented by an intersection graph of generalized polygons with one arc and one chord on a circle, because if a shape on a circle intersects the chord, it always intersects the arc.


\subsection {General Lower Bounds}\label{sec:generalLB}

In this section, we prove the following theorem.

\begin{theorem}\label{th:generalLB}
Consider a class of intersection graphs on a circle consisting of $n$ polygons,
each of which has at most 
$k = {\rm polylog}(n)$
chords or arcs.
Let $n_i$ be the number of 
all polygons on the circle with $i$ corners,
${\bar n}=(n_2, n_3, \ldots, n_k)$,
and
$N = \sum_{i=2}^{k} i \cdot n_i$.
Let $P_{n,k,{\bar n}}$ denote the total number of such graphs.
Then, the following holds:
\[
\log P_{n,k,{\bar n}} 
\geq \sum_{i=2}^{k} n_i \cdot i\log\frac{n}{i}
-n \log n
 -\Order(N \log\log n).
 \]
\end{theorem}

\begin{proof}
We count the number of graphs in the class of intersection graphs on a circle consisting of $n$ polygons,
each of which has at most $k$ chords and no arcs.
This gives a lower bound of $P_{n,k,{\bar n}}$.

Suppose that the circle polygon graph is given as a polygon circle representation with $n$ polygons on the circle.
We will consider partially-colored circle polygon graphs obtained from the following construction. 
Take $m \le n$ (to be determined) non-intersecting polygons
$A_1,\dots,A_m$ and paint $A_i$ with color $i$. 
Let the set of these $m$ polygons be $S$. 
For each of the remaining $n-m$ polygons, we will choose a subset of $S$, and for each such subset $X$ we will construct a polygon with $|X|$ corners such that distinct corners lie on distinct polygons from $X$. 
Note that 
each edge of such a polygon intersects with exactly two colored polygons.
This construction gives us a polygon-circle graph with $n$ vertices, where $m$ of these vertices are colored and
they form an independent set.
For $2 \le i \le k$, let $n_i$ and $m_i$ ($\le n_i$) be the number of 
all polygons and colored polygons on the circle with $i$ corners respectively. 
Similarly, let $N$ and $M$ be the number of total corners on the all polygons and colored polygons, respectively. 
From the definition, it is clear that $n = \sum_{i=2}^k n_i$, $m = \sum_{i=2}^k m_i$, 
$N = \sum_{i=2}^{k} i \cdot n_i$,  and $M = \sum_{i=2}^{k} i \cdot m_i$.
Let 
${\bar n}=(n_2, n_3, \ldots, n_k)$ and
${\bar m}=(m_2, m_3, \ldots, m_k)$.
Let us denote by $C_{n,k,{\bar n},{\bar m}}$ 
the number of such colored polygon-circle graphs,
and by $P_{n,k,{\bar n}}$ 
the number of polygon-circle graphs
for a given $k$ and ${\bar n}$.

We can first obtain an inequality
$\binom{n_2}{m_2} \binom{n_3}{m_3} \dots \binom{n_k}{m_k} \cdot m! \cdot P_{n,k,{\bar n}} \geq C_{n,k,{\bar n},{\bar m}}$
since every graph counted in $C_{n,k,{\bar n},{\bar m}}$ can be obtained by choosing and coloring $m_i$ polygons from $n_i$ polygons on its polygon circle representation of uncolored one for each $2 \le i \le k$.
Now we will find a lower bound for $C_{n,k,{\bar n},{\bar m}}$, which in turn will give a lower bound for $P_{n,k,{\bar n}}$.
Let us denote the collection of $i$-subsets of $S$ by $S_i$. 
Hence $|S_i| = {m \choose i}$. 
Also let $\mathcal{S}$ be a set of all possible $(k-1)$-tuples
$(Y_2, Y_3, \ldots, Y_k)$ where $Y_i$ is a subset of $S_i$ with $|Y_i| = n_i-m_i$.
Then,
$|\mathcal{S}| = \binom{\binom{m}{2}}{n_2-m_2} \binom{\binom{m}{3}}{n_3-m_3} \dots \binom{\binom{m}{k}}{n_k-m_k}$.
Now, the total number of graphs obtained by the above construction is at least $|\mathcal{S}|$ 
by the following observations:
\begin{enumerate}
\item[(i)] Each element in $\mathcal{S}$ defines at least one colored graph with $\sum_{i=2}^k (n_i-m_i) = n-m$ uncolored polygons (we might get more as the relative order of the corners of polygons within a colored polygon matters).
\item[(ii)] If $\cal T_1$ and $\cal T_2$ are elements of 
$\mathcal{S}$
and $\cal T_1\not= \cal T_2$, then 
the graphs corresponding to these two elements will be different. Basically, in the graphs obtained from this construction, uncolored $n-m$ vertices are distinguishable by only looking at their colored neighbors.
\end{enumerate}

\begin{figure}[t]
\begin{minipage}{.45\textwidth}
\centering
\begin{tikzpicture}[scale=0.4]
  \draw [line width=0.4mm] (0,0) circle (5cm);
  \draw [pink,line width=0.8mm] (10:5) arc [radius=5, start angle=10, end angle=70];
  \draw [green, line width=0.8mm] (80:5) arc [radius=5, start angle=80, end angle=140];
  \draw [yellow,line width=0.8mm] (150:5) arc [radius=5, start angle=150, end angle=210];
  \draw [red,line width=0.8mm] (220:5) arc [radius=5, start angle=220, end angle=280];
  \draw [Cyan,line width=0.8mm] (290:5) arc [radius=5, start angle=290, end angle=350];]
 
  \draw [fill=pink] (10:5)--(40:5)--(70:5)--(10:5) node at (40:5.3) {1};
  \draw [fill=green] (80:5)--(110:5)--(140:5)--(80:5) node at (110:5.3) {2};
   \draw [fill=yellow] (150:5)--(180:5)--(210:5)--(150:5) node at (180:5.3) {3};
  \draw [fill=red] (220:5)--(250:5)--(280:5)--(220:5) node at (250:5.3) {4};
  \draw [fill=Cyan] (290:5)--(320:5)--(350:5)--(290:5) node at (320:5.3) {5};
 
  \draw[thick] (160:5)--(300:5)--(230:5)--(160:5) node at (250:3.5) {$345$};
  \draw[thick] (170:5)-- (60:5)--(100:5)--(170:5) node at (100:3.7) {$123$};
  \draw[thick] (130:5)--(200:5)--(310:5)--(130:5) node at (180:1.5) {235};
 \end{tikzpicture}
\end{minipage}
\begin{minipage}{.45\textwidth}
\centering
\begin{tikzpicture}[scale=0.4]
\draw[fill=pink] (40:5) circle (10pt) node at (40:5) (1) {1};
\draw[fill=green] (110:5) circle (10pt) node at (110:5) (2) {2};
\draw[fill=yellow] (180:5) circle (10pt)node at (180:5) (3) {3};
\draw[fill=red] (250:5) circle (10pt)node at (250:5) (4) {4};
\draw[fill=Cyan] (320:5) circle (10pt) node at (320:5) (5){5};
\draw (90:2.5) circle (10pt) node at (90:2.5) (123){123};
\draw (180:2.5) circle (10pt) node at (180:2.5)  (235) {235};
\draw (270:2.5) circle (10pt) node at (270:2.5) (345) {345};

\draw (1)--(123)--(2);
\draw (123)--(3);
\draw (2)--(235)--(3);
\draw (5)--(235);
\draw (3)--(345)--(4);
\draw (5)--(345);
\draw (123)--(235)--(345)--(123);
\end{tikzpicture}
\end{minipage}
\caption{A realization of $G_{\cal T_1}$, where $G_{\cal T_1}=\{123, 235, 345\}$.}
\label{fig:cp_T1}
\end{figure}
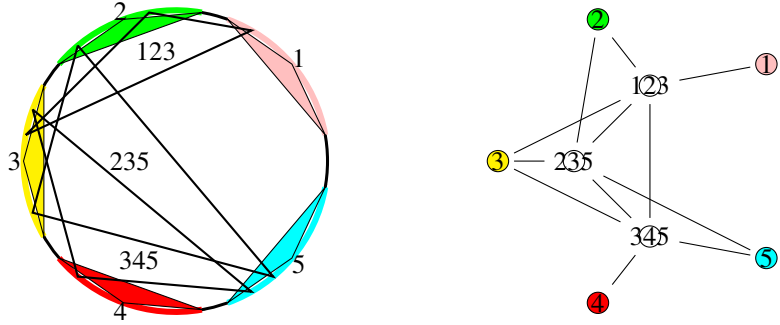
\begin{figure}[h]
\begin{minipage}{.45\textwidth}
\centering
\begin{tikzpicture}[scale=0.4]
  \draw [line width=0.4mm] (0,0) circle (5cm);
  \draw [pink,line width=0.8mm] (10:5) arc [radius=5, start angle=10, end angle=70];
  \draw [green, line width=0.8mm] (80:5) arc [radius=5, start angle=80, end angle=140];
  \draw [yellow,line width=0.8mm] (150:5) arc [radius=5, start angle=150, end angle=210];
  \draw [red,line width=0.8mm] (220:5) arc [radius=5, start angle=220, end angle=280];
  \draw [Cyan,line width=0.8mm] (290:5) arc [radius=5, start angle=290, end angle=350];]
 
  \draw [fill=pink] (10:5)--(40:5)--(70:5)--(10:5) node at (40:5.3) {1};
  \draw [fill=green] (80:5)--(110:5)--(140:5)--(80:5) node at (110:5.3) {2};
   \draw [fill=yellow] (150:5)--(180:5)--(210:5)--(150:5) node at (180:5.3) {3};
  \draw [fill=red] (220:5)--(250:5)--(280:5)--(220:5) node at (250:5.3) {4};
  \draw [fill=Cyan] (290:5)--(320:5)--(350:5)--(290:5) node at (320:5.3) {5};
 
  \draw[thick] (170:5)--(300:5)--(230:5)--(170:5) node at (250:3.5) {345};
  \draw[thick] (160:5)-- (60:5)--(100:5)--(160:5) node at (100:3.7) {123};
  \draw[thick] (130:5)--(200:5)--(310:5)--(130:5) node at (180:2.2) {235};
 \end{tikzpicture}
\end{minipage}
\begin{minipage}{.45\textwidth}
\centering
\begin{tikzpicture}[scale=0.4]
\draw[fill=pink] (40:5) circle (10pt) node at (40:5) (1) {1};
\draw[fill=green] (110:5) circle (10pt) node at (110:5) (2) {2};
\draw[fill=yellow] (180:5) circle (10pt)node at (180:5) (3) {3};
\draw[fill=red] (250:5) circle (10pt)node at (250:5) (4) {4};
\draw[fill=Cyan] (320:5) circle (10pt) node at (320:5) (5){5};
\draw (90:2.5) circle (10pt) node at (90:2.5) (123){123};
\draw (180:2.5) circle (10pt) node at (180:2.5)  (235) {235};
\draw (270:2.5) circle (10pt) node at (270:2.5) (345) {345};

\draw (1)--(123)--(2);
\draw (123)--(3);
\draw (2)--(235)--(3);
\draw (5)--(235);
\draw (3)--(345)--(4);
\draw (5)--(345);
\draw (123)--(235)--(345);
\end{tikzpicture}
\end{minipage}
\caption{Another realization of $G_{\cal T_1}$. Note that the edge between the vertices $123$ and $345$ is missing in this case}
\label{fig:cp_T2}
\end{figure}
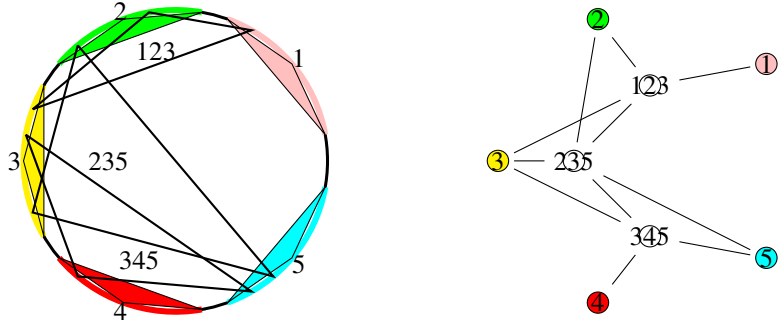
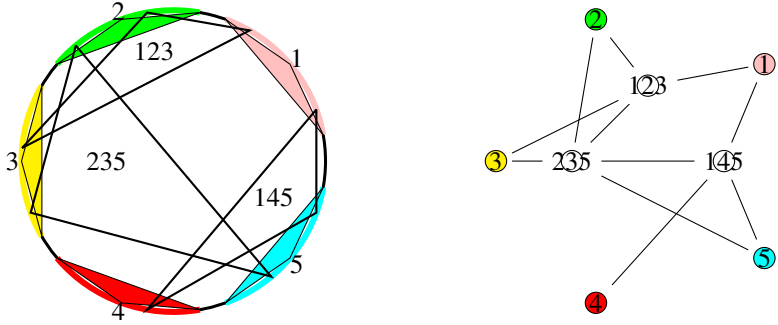
\begin{figure}[bt]
\begin{minipage}{.45\textwidth}
\centering
\begin{tikzpicture}[scale=0.4]
  \draw [line width=0.4mm] (0,0) circle (5cm);
  \draw [pink,line width=0.8mm] (10:5) arc [radius=5, start angle=10, end angle=70];
  \draw [green, line width=0.8mm] (80:5) arc [radius=5, start angle=80, end angle=140];
  \draw [yellow,line width=0.8mm] (150:5) arc [radius=5, start angle=150, end angle=210];
  \draw [red,line width=0.8mm] (220:5) arc [radius=5, start angle=220, end angle=280];
  \draw [Cyan,line width=0.8mm] (290:5) arc [radius=5, start angle=290, end angle=350];]
 
  \draw [fill=pink] (10:5)--(40:5)--(70:5)--(10:5) node at (40:5.3) {1};
  \draw [fill=green] (80:5)--(110:5)--(140:5)--(80:5) node at (110:5.3) {2};
   \draw [fill=yellow] (150:5)--(180:5)--(210:5)--(150:5) node at (180:5.3) {3};
  \draw [fill=red] (220:5)--(250:5)--(280:5)--(220:5) node at (250:5.3) {4};
  \draw [fill=Cyan] (290:5)--(320:5)--(350:5)--(290:5) node at (320:5.3) {5};
 
  \draw[thick] (20:5)--(340:5)--(260:5)--(20:5) node at (340:3.5) {145};
  \draw[thick] (175:5)-- (60:5)--(100:5)--(175:5) node at (100:3.7) {123};
  \draw[thick] (130:5)--(200:5)--(310:5)--(130:5) node at (180:2.2) {235};
 \end{tikzpicture}
\end{minipage}
\begin{minipage}{.45\textwidth}
\centering
\begin{tikzpicture}[scale=0.4]
\draw[fill=pink] (40:5) circle (10pt) node at (40:5) (1) {1};
\draw[fill=green] (110:5) circle (10pt) node at (110:5) (2) {2};
\draw[fill=yellow] (180:5) circle (10pt)node at (180:5) (3) {3};
\draw[fill=red] (250:5) circle (10pt)node at (250:5) (4) {4};
\draw[fill=Cyan] (320:5) circle (10pt) node at (320:5) (5){5};
\draw (90:2.5) circle (10pt) node at (90:2.5) (123){123};
\draw (180:2.5) circle (10pt) node at (180:2.5)  (235) {235};
\draw (0:2.5) circle (10pt) node at (0:2.5) (145) {145};

\draw (1)--(123)--(2);
\draw (123)--(3);
\draw (2)--(235)--(3);
\draw (5)--(235);
\draw (1)--(145)--(4);
\draw (5)--(145);
\draw (123)--(235)--(145);
\end{tikzpicture}
\end{minipage}
\caption{Some realization of $G_{\cal T_2}$, where $G_{\cal T_2}=\{123, 235, 145\}$.}
\label{fig:cpT3}
\end{figure}

\begin{example}~\label{example_cp}
Let $k=3$, $m=5$, and $n=8$. 
We consider triangle-circle graphs (the special case of generalized 3-polygon circle graph when $n_3 = n$, and $n_2 = 0$.) 
with 8 vertices, of which 5 are colored with $1,\dots,5$.
Then $S=\{1,2,3,4,5\}$ and 
$ \cal S=\{123,124,125,134,135,145,234,235,245,345\}$.
Here we simply write $xyz$ to denote the 2-tuple $(\emptyset, \{x,y,z\})$ for $1\le x,y,z \le 5$.
Any $3$-element subset of $\cal S$ will give us a graph.
For instance, let $\cal T_1=\{123, 235, 345\}$ and $\cal T_2=\{123, 235, 145\}$.
The graphs $G_{\cal T_1}$ and $G_{\cal T_2}$ corresponding to subsets $\cal T_1$ and $\cal T_2$, respectively, are different for the following simple reason: $G_{\cal T_1}$ has an uncolored vertex that has colored neighbors $3$, $4$, and $5$ but  $G_{\cal T_2}$ does not have such a vertex. (Compare Figure~\ref{fig:cp_T2} and \ref{fig:cpT3}.) On the other hand, more than one graph might be possible corresponding to only one subset $\cal T_1$ as shown in the Figure~\ref{fig:cp_T1} and \ref{fig:cp_T2}.
Both of them use the same $\cal T_1=\{123, 235, 345\}$.
The edges between colored vertices and uncolored vertices are the same but the set of edges between the uncolored vertices are different.
\end{example}

Now we obtain the lower bound of $\log {P_{n,k,{\bar n}}}$ as follows. From the above arguments, we obtain
$C_{n,k,{\bar n},{\bar m}} \geq \binom{\binom{m}{2}}{n_2-m_2} \binom{\binom{m}{3}}{n_3-m_3} \dots \binom{\binom{m}{k}}{n_k-m_k}$.
Combining with the upper bound of $C_{n,k,{\bar n},{\bar m}}$, we obtain
\[
\log{P_{n,k,{\bar n}}} \geq \sum_{i=2}^{k} {\log{\binom{\binom{m}{i}}{n_i-m_i}}} - \sum_{i=2}^{k} {\log{\binom{n_i}{m_i}}} - \log{m!}.
\]
We set $m=\frac{n}{\log{n}}$.
For each term on the right-hand side, the following inequalities hold (using the inequality ${a\choose b} \ge (a/b)^b$).
\begin{dmath*}
\sum_{i=2}^{k} {\log{\binom{\binom{m}{i}}{n_i-m_i}}}  
\geq \sum_{i=2}^{k} \log{ \left( \frac{\binom{m}{i}}{n_i-m_i} \right) ^ {n_i-m_i} } \\
= \sum_{i=2}^{k} { (n_i-m_i) \log{\binom{m}{i}} } - \sum_{i=2}^{k}  {(n_i-m_i) \log{ (n_i-m_i)} } 
\geq \sum_{i=2}^{k} (n_i-m_i) i\log\frac{m}{i} -n \log n \\
= \sum_{i=2}^{k} (n_i-m_i) i\left(\log\frac{n}{i}-\log\log n\right) -n \log n \\
= \sum_{i=2}^{k} n_i \cdot i\log\frac{n}{i}
 -\sum_{i=2}^{k} m_i \cdot i\log\frac{n}{i}
-N \log\log n -n \log n \\
\geq \sum_{i=2}^{k} n_i \cdot i\log\frac{n}{i}
 -M \log n -N \log\log n -n \log n \\
\end{dmath*}
\[
\sum_{i=2}^{k} \log{\binom{n_i}{m_i}} 
\leq \sum_{i=2}^{k} m_i \log{n}
= m \log{n} \leq n
\]
Therefore, 
\begin{dmath*}
\log{P_{n,k,{\bar n}}} 
\geq \sum_{i=2}^{k} n_i \cdot i\log\frac{n}{i}
 -M \log n -N (\log k+\log\log n) -n \log n
 - \Order(n)
\end{dmath*}

To satisfy $M \le N/\log n$, 
we choose (and color) $m$ polygons as follows.
For $1 \le j \le n$, let $d_j$ be the number of corners of $j$-th polygon in the representation. 
Without loss of generality, we order the polygons to satisfy $d_1 \le d_2 \le \cdots \le d_n$.
Now we claim that $M \le N/\log n$ if we choose first $m$ polygons to be colored.
To prove the claim, suppose $d_{m+1} \ge N/n$. Then $\sum_{j = m+1}^{n} d_j \ge (n-m) \cdot N/n = N(1-1/\log n)$, 
which implies $M = \sum_{j = 1}^{m} d_j \le N/\log n$.
Next, suppose $d_{m+1} < N/n$. In this case, $M \le N/n \cdot m = N/\log n$, which proves the claim.
From the assumption $k = {\rm polylog}(n)$,
$N(\log k + \log \log n) = \Order(N\log\log n)$.
Thus, 
$\qquad \log{P_{n,k,{\bar n}}}
\geq \sum_{i=2}^{k} n_i \cdot i\log\frac{n}{i}
-n \log n
-\Order(N \log\log n).
$
\end{proof}


\begin{corollary}\label{cor:1}
Space lower bounds for 
circle-trapezoid graphs 
and circular-arc graphs
with $n$ vertices
are $3n \log n - \Order(n \log\log n)$ bits
and $n \log n - \Order(n \log\log n)$ bits,
respectively.
\end{corollary}

\begin{proof}

For circle-trapezoid graphs, 
in the proof of Theorem~\ref{th:generalLB},
we create $m = n/\log n$ colored circle-trapezoids,
consisting of two arcs and two chords, and place them
on the circle so that they do not overlap each other.
Then we place $n-m$ uncolored circle-trapezoids with
two arcs and two chords such that
each chord intersects with exactly two colored circle-trapezoids.
The theorem gives a lower bound for the number of
such graphs, which is a lower bound of the number of
circle-trapezoid graphs.

For circular-arc graphs,
instead of a circle-trapezoid, we consider
a $2$-polygon with one arc and one chord.
Then we obtain the desired bound.
\end{proof}

Note that to obtain a space lower bound for 
trapezoid graphs, we cannot use Theorem~\ref{th:generalLB}
because we cannot distinguish the upper line and 
the lower line. Another proof is given in
Section~\ref{sec:trapezoidlb}.

\subsection{A Succinct Representation} \label{sec:generaluuper}
Now we provide a succinct representation for generalized circle polygon
graph $G$ with $n$ generalized polygons on a circle.
Let $N$ be the total number of corners of the polygons.

Note that the recognition algorithm of general $k$-polygon-circle graphs,
which is a sub-class of the generalized circle polygon graphs,
is NP-complete~\cite{DBLP:conf/gd/KratochvilP03}. 
Thus we assume that $G$ is given as a polygon-circle representation with $n$ polygons,
which is defined (for a graph $G = (V,E)$) as a mapping $\mathcal{P}$ of vertices in $V$ to polygons inscribed into a circle such that $(u, v) \in E$ if and only if $\mathcal{P}(u)$
intersects $\mathcal{P}(v)$. 

Then, a {\it corner-string} of a polygon-circle representation is a string produced by starting at any arbitrary location on the circle, and proceeding around the circle in clockwise order, adding a label
denoting the vertex represented by a polygon each time a corner of a polygon
encountered (denoted by the array $S$ in Figure~\ref{figure:circle_trapezoid}). Note that a single polygon-circle representation has many possible corner-strings, depending on the starting point. As the naive encoding of  $S$ uses 
$N \ceil{\log n}$ 
bits, it is not succinct, and does not support efficient queries.
Therefore we convert $S$ into another representation and add auxiliary data structures
for efficient queries.  First, we convert $S$ into a bit array $F$ of length $N$
and another integer array $S^\prime$ of length 
$N-n$.
The entry $F[i]$ is $1$ if $S[i]$ is the first
occurrence of the value in $S$, and $0$ otherwise.
The array $S^\prime$ stores all entries of $S$
except for the first occurrence of each value
in the same order as in $S$.
We store $F$ using the data structure of Lemma~\ref{rankselect2},
and $S^\prime$ using the data structure of Lemma~\ref{rankselect}.
Then the space becomes
$(N-n)\log n + \Order(N\log n/\log\log n))$
bits, which is succinct.  Using $F$ and $S^\prime$,
we show how to support $\access{}(i, S)$ and $\rank_{\alpha}(i, S)$
in $\Order(\log\log n)$ time
and $\select_{\alpha}(i, S)$ in $\Order(1)$ time.
\begin{itemize}
    \item $\access{}(i, S) = \begin{cases}
                  \rank_{1}(i, F)                    & (\access{}(i, F)=1) \\
                  \access{}(\rank_{0}(i, F), S^\prime) & ({\rm otherwise})
                  \end{cases}$
    \item $\rank_{\alpha}(i, S) = 
                \rank_{\alpha}(\rank_{0}(i,F), S^{'}) +
                    \begin{cases}
                  1 & (\rank_{1}(i,F) \geq S[i]) \\
                  0 & ({\rm otherwise})
                  \end{cases}$
    \item $\select_{\alpha}(i, S) = \begin{cases}
                  \select_{1}(\alpha, F) & (i=1) \\
                  \select_{0}(\select_{\alpha}(i-1, S^\prime), F) & ({\rm otherwise})
                  \end{cases}$
\end{itemize}

We can regard as if the array $S$
were stored and access to $S$ were done in $\Order(\log\log n)$
time. 

Now we show the space bound of the representation.
We compress the corner-string $S$, in which 
a character $2 \le i \le k$ appears $n_i$ times in $S$.
The length of $S$ is $N = \sum_{i=2}^k n_i \cdot i$.
For each character $i$, its first occurrence in $S$ is encoded
in a bit-vector of length $N$.
Other characters are stored in a string $S^\prime$ of length $N-n$.
Each character $i$ appears $n_i-1$ times in $S^\prime$.
We compress $S^\prime$ into its order-$0$ entropy.
Then the total space is
\begin{dmath*}
\sum_{i=2}^{k} n_i (i-1) \log\frac{N-n}{i-1} + \Order(N) 
\leq \sum_{i=2}^{k} n_i (i-1) \log\frac{nk}{i/2} + \Order(N) \\
\leq \sum_{i=2}^{k} n_i \cdot i \log\frac{n}{i}
 - n \log n + \Order(N \log k).
\end{dmath*}
%
If $k = \order(\log n/\log \log n)$, the lower bound of Theorem~\ref{th:generalLB} is
\begin{dmath*}
\log{P_{n,k,{\bar n}}}
\geq \sum_{i=2}^{k} n_i \cdot i\log\frac{n}{i}
-n \log n
 -\Order(N \log\log n)\\
\geq \sum_{i=2}^{k} n_i \cdot i\log\frac{n}{i}
-n \log n -\order(n \log n).
\end{dmath*}
On the other hand, the upper bound is
\begin{dmath*}
\sum_{i=2}^{k} n_i \cdot i \log\frac{n}{i}
 - n \log n + \Order(N \log k)
 \leq \sum_{i=2}^{k} n_i \cdot i \log\frac{n}{i}
 - n \log n + \order(n \log n).
\end{dmath*}
Therefore this upper bound matches the lower bound, up to lower order terms.

\subsection{Query Algorithms}\label{sec:query}
We now describe how to support navigation queries using the data structure of Section~\ref{sec:generaluuper}.
Consider a vertex $u$ in $G$.
Assume the vertex $u$ corresponds to a $k$-polygon, which is represented by $k$ many integers $u$ in $S$.
The polygon has $k$ edges, for $i$-th edge
($1 \le i \le k-1$), we consider an interval of $S$
between $i$-th occurrence of $u$ and $(i+1)$-st
occurrence of $u$, that is,
$[\select_{u}(i, S),\select_{u}(i+1, S)]$.
Let $I(u, i)$ denote this interval.
For $k$-th edge, the interval becomes the union of
$[\select_{u}(k, S),N]$
and $[1,\select_{u}(1, S)]$.

Consider two polygons $u$ and $v$.
We check for each side $e$ of $u$ if $e$ intersects with a side $f$ of $v$.
Let $I(u, i) = [\ell, r]$ be the interval of $e$
and $I(v, j) = [s,t]$ be the interval of $f$.
There are four cases.
(1) $e$ is a chord and $f$ is a chord.
Then $e$ and $f$ intersect iff 
$[\ell, r] \cap [s, t] \neq \emptyset$, $[s,t] \not\subset [\ell, r]$, and $[\ell, r] \not\subset [s, t]$. 
(2) $e$ is an arc and $f$ is a chord.
This case is the same as (1) in addition the case when $[s,t] \subset [\ell, r]$.
(3) $e$ is a chord and $f$ is an arc.
This case is the same as (1) in addition to the case when $[\ell,r] \subset [s, t]$.
(4) $e$ is an arc and $f$ is an arc.
Then $e$ and $f$ intersect
iff $[\ell, r] \cap [s, t] \neq \emptyset$.

We add new data structures $N$, $N_a$, $P_a$, $N_c$, $P_c$, and $A$
defined as follows.
Let $I(u, i) = [\ell_i, r_i]$ denote the $i$-th interval of $S$, 
defined above, and $d_u$ be the number of corners of $u$.
Then $A$ is a bit array of length $N$ where $A[\ell_i] = 1$ 
if and only if $[\ell_i, r_i]$ corresponds to an arc of $u$.
The arrays $N$, $N_a$, $P_a$, $N_c$, and $P_c$ are defined as follows
where $u = S[i]$.
\begin{align*}
    N[i] &= \begin{cases}
            \select_{u}(\rank_{u}(i, S)+1,S) & (\mbox{if $\rank_{u}(i, S) < d_{u}$}) \\
            \infty & (\mbox{otherwise})
            \end{cases}\\
    N_a[i] &= \begin{cases}
            \select_{u}(\rank_{u}(i, S)+1,S) & (\mbox{if $A[i] = 1$
            and $\rank_{u}(i, S) < d_{u}$}) \\
            \infty & (\mbox{if $A[i] = 1$
            and $\rank_{u}(i, S) = d_{u}$}) \\
            0  & (\mbox{otherwise})
            \end{cases}\\
    P_a[i] &= \begin{cases}
            \select_{u}(\rank_{u}(i, S)-1,S) & (\mbox{if the side ending at $S[i]$ is an arc
            and $\rank_{u}(i, S) > 1$}) \\
            0 & (\mbox{if the side ending at $S[i]$ is an arc
            and $\rank_{u}(i, S) = 1$}) \\
            \infty  & (\mbox{otherwise})
            \end{cases}\\
    N_c[i] &= \begin{cases}
            \select_{u}(\rank_{u}(i, S)+1,S) & (\mbox{if $A[i] = 0$
            and $\rank_{u}(i, S) < d_{u}$}) \\
            0  & (\mbox{otherwise})
            \end{cases}\\
    P_c[i] &= \begin{cases}
            \select_{u}(\rank_{u}(i, S)-1,S) & (\mbox{if the side ending at $S[i]$ is a chord
            and $\rank_{u}(i, S) > 1$}) \\
            \infty  & (\mbox{otherwise})
            \end{cases}
\end{align*}
The array $N_a[i]$ ($N_c[i]$) stores the other endpoint of an arc (a chord)
starting from $S[i]$. Similarly, the array $P_a[i]$ ($P_c[i]$) stores the other endpoint of an arc (a chord) ending with $S[i]$. Finally, the array $N[i]$ stores the other endpoint of an arc or a chord starting from $S[i]$.
The difference between $N_a$ and $N_c$ is that in $N_c$, we do not store the last side of a polygon.
We do not store these arrays explicitly;
we store only the range maximum data structures for $N$, $N_a$ and $N_c$,
and the range minimum data structures for $P_a$ and $P_c$.
We can obtain any entry of these arrays in $\Order(\log\log n)$ time
using the above formula.

\begin{lem}\label{lem:k-polygon-circle1}
If generalized polygons $u$ and $v$ intersect,
there exists a side $e$ of $u$ with interval $[\ell,r]$
and a side $f$ of $v$ with interval $[s, t]$ satisfying
at least one of the following.
\begin{enumerate}
\item[(1)] Both $e$ and $f$ are chords, 
neither $e$ or $f$ is the last side, 
and $(\ell < s < r$ and $r < t (= N_c[s]))$
or $(\ell < t < r$ and $\ell > s (=P_c[t]))$.
\item[(2)] $e$ is an arc and $f$ is a chord, 
$f$ is not the last side, and 
 $\ell < t(=N_a[s]) < r$ or $\ell < s(=P_a[t]) < r$.
\item[(3)] $e$ is a chord and $f$ is an arc, 
$e$ is not the last side, 
and $s < r< t(=N_a[s])$ or $s(=P_a[t]) < \ell < t$. 
\item[(4)] Both $e$ and $f$ are arcs, and
$s < r$ and $\ell < t (= N_a[s])$.
\end{enumerate}
\end{lem}
Note that for an arc $e$ that is the last side of $u$,
the interval is divided into two.
We regard as if $e$ is divided into two arcs
and apply the lemma to each of them.

\begin{figure}
\begin{center}
\includegraphics[scale=0.30]{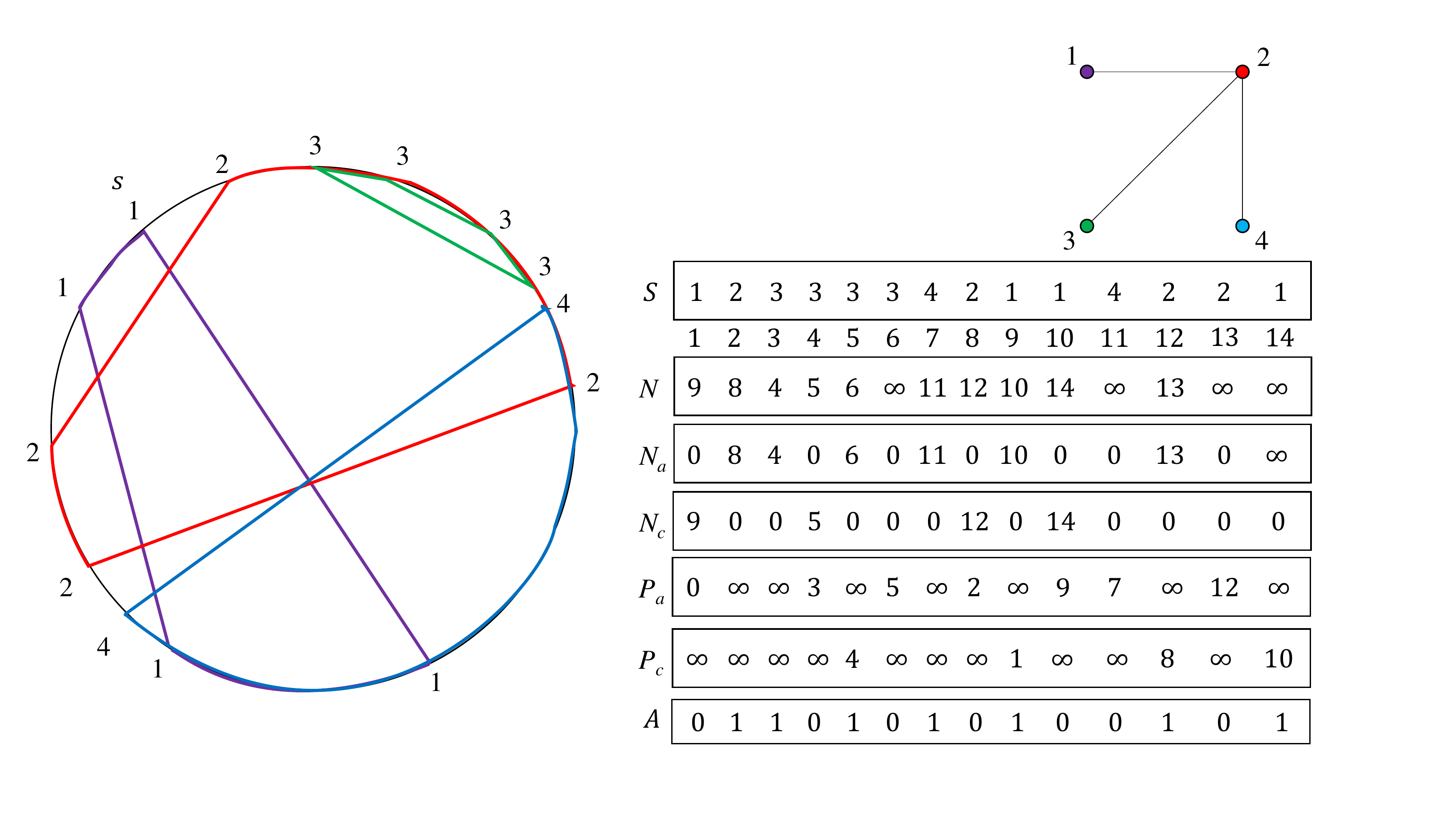}
\end{center}
\caption{A circle-trapezoid graph with $n=4$ given as polygon-circle representation (left), and graph form (top-right).  The array $S$ encodes the corner-string starting from the marked location $s$ on the circle. The arrays $S, N_a, P_a, N_c, P_c$ are not stored; we store only $F$, $S^\prime$, $A$, range max data structures for $N, N_a, N_c$, and range min data structures for $P_a, P_c$.  In the figure we omit $F$ and $S^\prime$.}
\label{figure:circle_trapezoid}
\end{figure}

Figure~\ref{figure:circle_trapezoid} shows an example of our representation.
For a chord of polygon $2$ whose interval is $[8,12]$,
polygons $1$ and $4$ intersect with the chord because
$N_c[10]=14 > 12$ and $P_c[9] = 1 < 8$.
For a chord of polygon $3$ whose interval is $[5,6]$,
polygon $2$ intersects with the chord because
$N_a[2]=8 > 6$.
For an arc of polygon $1$ whose interval is $[9,10]$,
polygons $4$ intersects with the arc because
$N_a[7]=11 > 9$.

Using this idea, we obtain an algorithm
for \adjacent{} query.
\\
\noindent\textbf{$\adjacent{}\boldsymbol{(u, v)}$ query: } Consider the intervals $I(u, 1)$, $I(u, 2), \ldots$, $I(u, d_u)$
and $I(v, 1)$, $I(v, 2), \ldots$, $I(v, d_v)$.
We scan these intervals in the clockwise order on the circle,
and for each endpoint of an interval, we check the condition of Lemma~\ref{lem:k-polygon-circle1}.
For each interval, checking this condition takes $O(\log\log n)$ time, and since we need to check at most $d_u + d_v$ intervals, the time complexity is $O(k \log\log n)$.

Next we consider $\neighbor{}\boldsymbol{(u)}$ query.
For each side (chord or arc) of $u$, we want to enumerate all generalized polygons $v$
satisfying the conditions of Lemma~\ref{lem:k-polygon-circle1}.
For each chord $e$ of $u$, we can find all chords which intersects with $u$
as follows.  Let $[\ell,r]$ be the interval of $e$.
First we obtain $m = \rMq(N_c, \ell, r)$.
If $N_c[m] \le r$, all entries of $N_c$ in $[\ell, r]$
are less than $r$, and there are no polygons
intersecting $e$.  Therefore we stop enumeration.
If $N_c[m] > r$, the polygon $S[m]$ intersects with $u$.
To check if there is another such polygon,
we recursively search for $[\ell, m-1]$ and $[m+1,r]$.
The time complexity is $\Order(d \log\log n)$
where $d$ is the number of entries $m$ such that
$N_c[m] > r$.
We also process $P_c$ analogously.


For an arc $e$ of $u$, we can enumerate all chords of the other
generalized polygons which intersects with $e$ is obtained by finding all $S[m]$ such that $\ell < m < r$.  Such distinct $m$ can be obtained
by finding all $m$ such that (i) $\ell < m < r$, and (ii) $\ell < N_c[m]$ or $P_c[m] < r$
using the range maximum data structure.

For a chord $e$ of $u$, we can enumerate all arcs of other generalized polygons
which intersects with $e$ is obtained by finding all $S[m]$ such that
$1 \le m < r$ and $N_a[m] > r$, or
$\ell < m \le N$ and $P_a[m] < \ell$.

For an arc $e$ of $u$, we can enumerate all arcs of other generalized polygons
which intersects with $e$ is obtained by finding all $S[m]$ such that
$1 \le m < r$ and $N_a[m] > \ell$, or
$\ell < m \le N$ and $P_a[m] < r$.

\noindent\textbf{$\neighbor{}\boldsymbol{(u)}$ query: }For
each interval $I(u, i)=[\ell, r]$, we output all polygons
$S[m]$ satisfying one of the above conditions.
However there may exist duplicates.  To avoid
outputting the same polygon twice, we use a bit array
$D[1\ldots,n]$ to mark which polygon is already output.
The bit array is initialized by $0$ when we create
the data structure.  At a query process,
before outputting a polygon $v$,
we check if $D[v]=1$.  If it is, $v$ is already output
and we do not output again.  If not, we output $v$
and set $D[v]=1$.  After processing all intervals
of $u$, we have to clean $D$.  To do so, we run
the same algorithm again.  But this time we output
nothing and set $D[v]=0$ for all $v$ found by the algorithm.  The time complexity is 
$\Order(|\degree{}(v)| \cdot k\log\log n)$
where $k$ is the maximum number of sides in each generalized polygon.
\\
\noindent\textbf{$\degree{}\boldsymbol{(v)}$ query: } The  $\degree{}\boldsymbol{(v)}$ can be answered by returning
the size of the output of the $\neighbor{}\boldsymbol{(u)}$ query, in
$O(|\degree{}(v)| \cdot k\log\log n)$ time. 
Note that by adding an integer array of length $n$
storing the degree of each vertex explicitly,
$\degree{}\boldsymbol{(v)}$ can be supported in $O(1)$ time. The whole data structure is still succinct
if 
$k = \omega(1)$, 
but it is not if $k = O(1)$.

Finally, we show how one can represent various classes of intersection graphs
by our representation.  Generalized polygons in each class are represented as follows.
\begin{itemize}
    \item $k$-polygon-circle: set all $A[i]=0$ for all $i$ (all sides are chords).
    \item circle-trapezoid: the number of sides is $4$ and arcs and chords appear alternately.
    \item trapezoid: we split a circle in half equally (upper and lower part), and both upper and lower part have 2 corners. Now arcs and chords appear alternately, from the arcs on the upper part.
    \item circle and permutation: the number of sides is $2$ and all sides are chords.
    \item circular-arc and interval: the number of sides is $2$ and there are an arc and a chord.
    Set all the entries of $N_c$ to be $0$ and $P_c$ to be $\infty$ so that
    the query algorithms do not output any chord.
\end{itemize}


Thus, we obtain the following result.
\begin{theorem}\label{thm:generaluppernew}
Consider an intersection graph of $n$ (generalized) polygons on a circle.
Let $n_i$ be the number of 
all polygons on the circle with $i$ corners 
($2 \le i \le k$),
where $k$ is the maximum number of corners among the polygons on a circle,
and $N$ be the total number of corners of the polygons.
There exist a 
$\left(\sum_{i=2}^{k} n_i \cdot i \log\frac{n}{i}
 - n \log n + \Order(N \log k) \right)$-bit
representation of the graph
that can support
$\adjacent{}(u, v)$ query in $\Order(k\log\log n)$ time, and
$\neighbor{}(v)$ and $\degree{}(v)$ queries in $\Order(k|\degree{}(v)| \cdot \log\log n)$ time.
Also, the representation is succinct 
(i.e., matches the lower bound of Theorem~\ref{th:generalLB} to within lower order terms)
when $k = \order(\log n/\log \log n)$.
\end{theorem}
\begin{corollary}\label{cor:simple}
For an intersection graph of $n$ (generalized) polygons
with at most $k$ corners
on a circle, let $N$ be the total number of corners of the polygons.
There exist an $((N-n)\log n + \Order(N \log k))$-bit representation of the graph.
For any $k$-polygon-circle graph, there exists
a $((k-1)n\log n + \Order(nk \log k))$-bit representation.
\end{corollary}

\section{Circle graphs}\label{sec:circle}
In this section, we first show that for circle graphs, the lower bound of Theorem~\ref{th:generalLB} can be improved to $n\log n-\Order(n)$ bits. Next, we give an alternative succinct representation of circle graphs, which can answer $\degree{}(v)$ queries independent of $|\degree{}(v)|$, but takes more time for the other two queries compared to the representation of Theorem~\ref{thm:generaluppernew}.
\subsection{Lower bound}
In this section we show that $\log C_n \ge n\log n-O(n)$ as $n \to \infty$, where $C_n$ is the number of unlabeled circle graphs with $n$ vertices.
We first take a circle with $2n$ equally spaced points on it, and label the points 1 to $2n$ clockwise such that the first $n$ points lie on the upper semicircle and the rest lie on the lower semicircle. These $2n$ points will be the endpoints of $n$ disjoint chords. First, on each semicircle, we take $k$ chords, each of which determines an arc with $\ell$ points on it, excluding the endpoints of the chord. So the first chord on the upper semicircle will connect the points 1 and $\ell+2$, the next chord will connect the points $\ell+3$ and $2\ell+4$, and etc. We will call these chords \emph{special} chords. Now color these special chords with the colors 1 through $2k$ in the canonical order, i.e., in the order we see them when we traverse the circle clockwise starting from point 1. So far, we have used $4k$ out of $2n$ points, and the remaining $2n-4k$ points lie on $2k$ arcs determined by the $2k$ special chords. 
Let these arcs be $A_1,\dots,A_{k}$ (for the upper semicircle) and $A_{k+1},\dots,A_{2k}$ (for the lower semicircle) (see Figure~\ref{figure:lb_circle1} for an example).
Since we want $\ell$ to satisfy $2k\ell+4k=2n$, $\ell$ is defined to $\frac{n-2k}{k}$.
\begin{figure}
\begin{center}
\includegraphics[scale=0.4, bb=0 0 459 451]{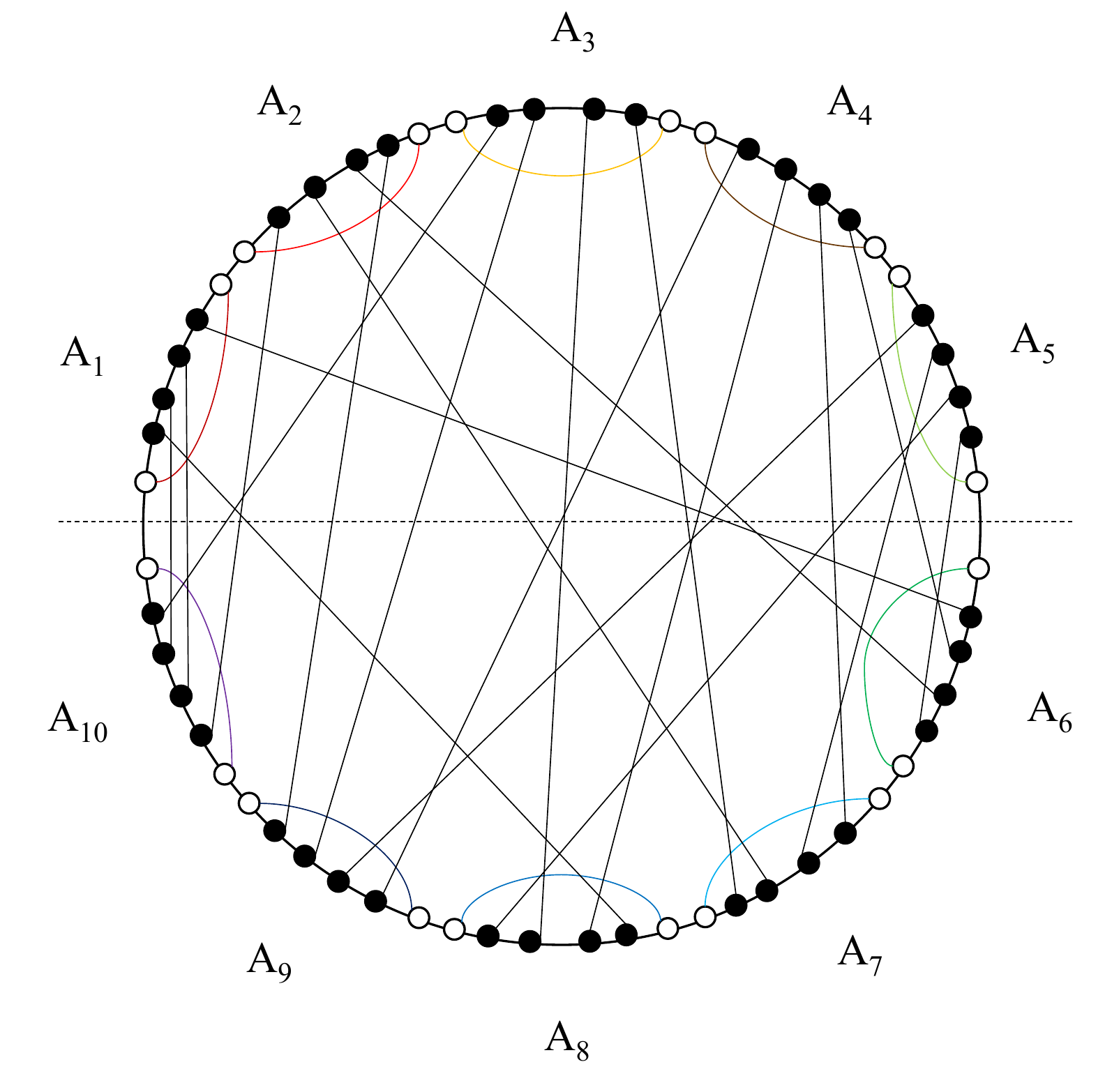}
\end{center}
\caption{$A_1,\dots,A_{k+1},\dots,A_{2k}$ when $k = 5$ and $\ell = 4$}
\label{figure:lb_circle1}
\end{figure}
%
From now on we will assume $k$ and $\ell$ are integers where $2k\ell+4k=2n$. (When they are not integers, we can easily modify the proof by taking appropriate ceilings and/or floors.)

Now what we want is to match the unused $k\ell$ points on the upper semicircle with the $k\ell$ unused points on the lower semicircle. The number of such matchings is $(k\ell)!$.
For each pair in the matching, if we draw the chord connecting the points in the pair, we get $n$ chords which gives a colored circle graph. (Each chord corresponds to a vertex.)
The $2k$ vertices corresponding to the special $2k$ chords are colored 1 through $2k$ (in the same canonical order), and the other $n-2k$ vertices are uncolored.

Let $M$ be a matching from $\cup_{i=1}^k A_i$ to $\cup_{j=k+1}^{2k} A_j$. We call $M$ a \textit{bad matching} if it contains a triple of pairs $((x_1,y_1), (x_2,y_2), (x_3,y_3))$ such that $x_1,x_2,x_3$ lie on $A_i$ for some $i\le k$ and $y_1,y_2,y_3$ lie on $A_{k+j}$ for some $j\le k$. Otherwise we call it a \textit{good matching} (denoted by $\cal M$). 
Then, the following lemma shows that when $n \to \infty$, almost all matchings are good. 

\begin{lemma}
Let $k= n^{3/4+\eps}$ for some fixed small $\eps>0$. 
For a random matching $M$, the expected number of triples of pairs $((x_1,y_1), (x_2,y_2), (x_3,y_3))$ in bad matching tends to 0 as $n \to \infty$. i.e,
$\frac{|\cal M|}{(k\ell)!}= 1-o(1)$
as $n\to \infty$. Consequently, almost all matchings are good.
\end{lemma}

\begin{proof}
Let $k$ and $M$ be as above and let $X$ denote the number of triplets $((x_1,y_1), (x_2,y_2), (x_3,y_3))$ in $M$ such that $x_1,x_2,x_3$ lie on $A_i$ for some $i\le k$ and $y_1,y_2,y_3$ lie on some arc $A_{k+j}$ for some $1\le j\le k$. 
Let  $X_{i,j}$ denote the number of triplets from $A_i$ to $A_{k+j}$ for $1\le i,j\le k$. Hence we have
$X=\sum_{i=1}^k\sum_{j=1}^k X_{i,j}$. Also
Letting $\mu=\mean{X}$ and $\mu_{i,j}=\mean{X_{i,j}}$, we have, by the linearity of the expectation, 
$\mu= \sum_{i=1}^k\sum_{j=1}^k \mu_{i,j}$.

In order to compute $\mu_{i,j}$, we take three points from each of $A_i$ and $A_{k+j}$ in ${\ell \choose 3}{\ell \choose 3}$ ways, and multiply this number by the probability that the first set of three points are matched to the second set of three points, which is $\frac{3!}{(k\ell)(k\ell-1)(k\ell-2)}$. Hence
\[
\mu_{i,j} = {\ell \choose 3}{\ell\choose 3}\, \frac{(k\ell-3)!3!}{(k\ell)!} = O\left( \frac{\ell^3}{k^3}\right).
\]
Since there are $k^2$ summands in the double sum, we get
\[
\mu=\sum_{i=1}^k\sum_{j=1}^k \mu_{i,j} = k^2O\left( \frac{\ell^3}{k^3}\right)= O(\ell^3/k)= O(n^3/k^4)\to 0.
\]
By Markov's inequality, we have
$\pr{X>0}=\pr{X\ge 1} \le \mean{X} =\mu =o(1)$.
Hence 
$\frac{|\cal M|}{(k\ell)!}= \pr{X=0}= 1-\pr{X\ge 1} \ge 1-o(1)$, 
as desired.
\end{proof}

\begin{figure}
\begin{center}
\includegraphics[scale=0.4, bb=200 150 800 400]{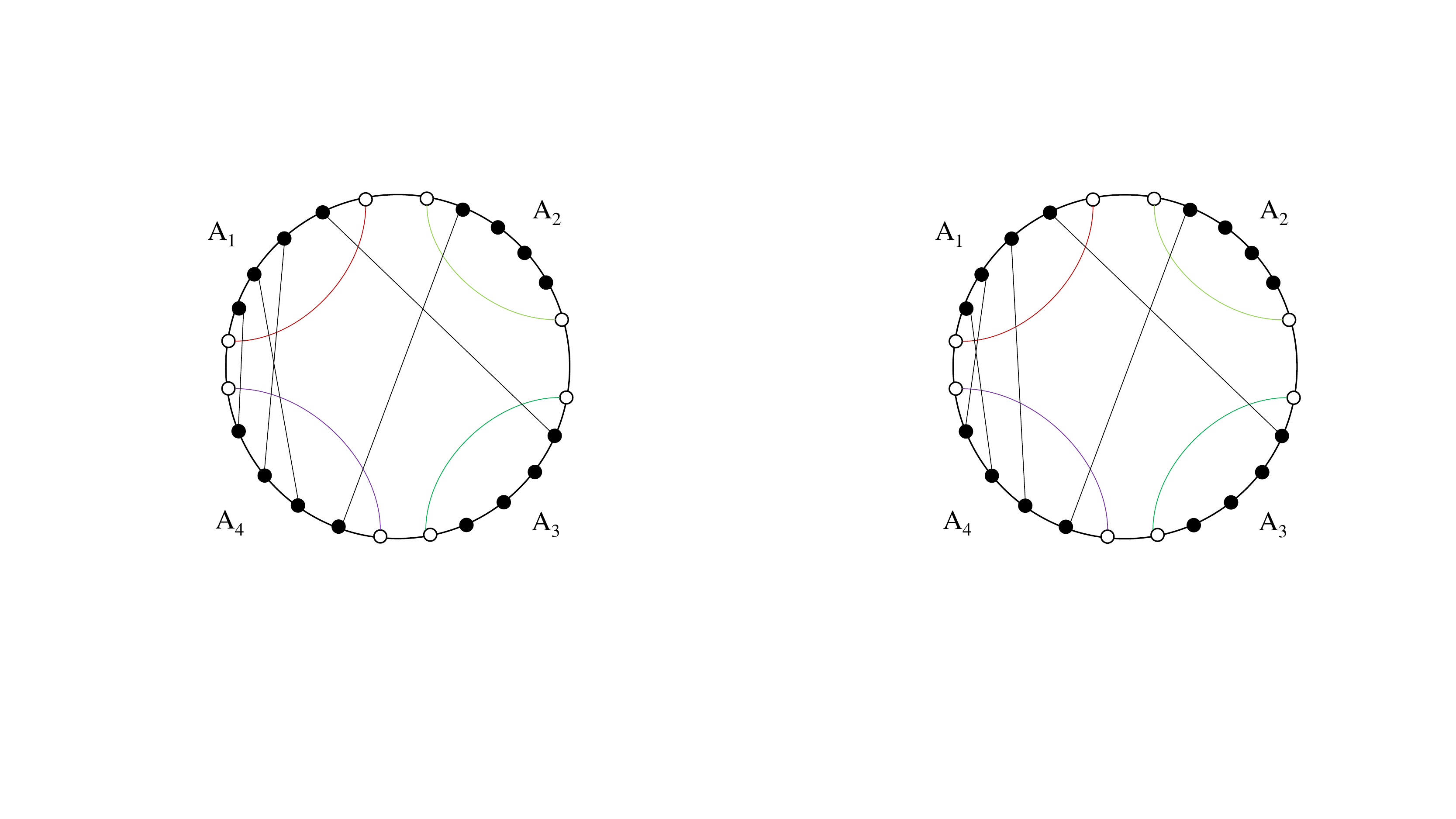}
\end{center}
\caption{Two different diagrams with the same graphs. (Suppose the unmatched points in the two diagrams are matched in the same way.)}
\label{figure:lb_circle2}
\end{figure}
Now we prove $\log C_n \ge n\log n-O(n)$, as $n\to \infty$. Let $k= n^{3/4+\eps}$. By the previous lemma, we have $|\cal M|= (1-o(1))(k\ell)!$.
The nice thing about a good matching is that it can be recovered from its (colored) circle graph. (In other words, there is a one to one matching between the set of good matchings and their corresponding graphs. Figure~\ref{figure:lb_circle2} shows that for bad matchings, this property does not hold) To see this, we first note that each colored vertex is unique; if we see a color $i$ on a vertex, that vertex corresponds to the $i$-th special chord. Next, each uncolored vertex has exactly two colored neighbors, which gives us the ability to determine the two arcs its endpoints lie on; if these two neighbors are colored $i$ and $j$, then the arcs containing the endpoints of the chord are $A_i$ and $A_j$. Finally, suppose two uncolored vertices $u$ and $v$ have a common neighbor with color $i$. Then one endpoint of each of the chords corresponding to $u$ and $v$ lie on $A_i$, and the relative order of these endpoints is determined by whether the vertices $u$ and $v$ are adjacent or not.
Hence
\begin{align*}
C_n {n \choose 2k}(2k)! & \ge \text{number of circle graphs with $2k$ colored vertices}\\
&\ge \text{number of circle graphs obtained from the construction above}\\
&\ge |\cal M| = (1-o(1))(k\ell)!= (1-o(1))(n-2k)!,
\end{align*}
where the last inequality follows from the fact that there is a one to one matching between the set of good matchings and their corresponding graphs.
Using
${n \choose 2k}(2k)! \le n^{2k}$,
we can write the above inequality as
$C_n n^{2k}\ge (1-o(1))(n-2k)!$.
Taking the logarithms and using the Stirling's approximation gives
$\log C_n = n\log n- O(n)$.		

We summarize the result in the following theorem.
\begin{theorem}\label{thm:lb_circle}
For unlabeled circle graph $G$ with $n$ vertices. Then at least $n\log n-\Order(n)$ bits are necessary to represent $G$.
\end{theorem}

\subsection{Alternative succinct representation}\label{sec:ubcircle}
In this section, we give alternative succinct representation of circle graphs. 
Before describing the representation, we first introduce the \textit{orthogonal range queries on grids} which defined as follows. Given a set $P$ of $n$ points on an $n \times n$ 2-dimensional grid, the orthogonal range queries on grids are consist of the following queries:
\begin{itemize}
    \item $\cnt{}(P, R)$ : returns the number of points of $P$ within the rectangular range $R$.
    \item $\report{}(P, R)$ : reports the points of $P$ within the rectangular range $R$
\end{itemize}

The following lemma shows that there exists a succinct representation to support these queries efficiently, which we used in our representation.

\begin{lem}[\cite{DBLP:conf/wads/BoseHMM09}]\label{2drange}
Given a set of $n$ points $P$ on an $n \times n$ grid, there exists an $n \log n + o(n \log n)$-bit data structure, such that for any $(x, y) \in P$ and the rectangular range $R$, one can answer (i) $\cnt{}(P, R)$ queries in $O(\log n / \log {\log n})$ time, and (ii) $\report{}(P, R)$ queries in $O(k\log n / \log {\log n})$ time, where $k$ is the size of the output.
\end{lem}
\begin{remark2}
When there exists no two points $p_1, p_2 \in P$ where $p_1 = (x, y_1)$ and $p_2 = (x, y_2)$ with some $1 \le x \le n$ and $1 \le y_1, y_2 \le n$, we define the query $Y(x, P)$ as $\report{}(P, [x, x] \times [1, n])$. Since the size of $\report{}(P, [x, x] \times [1, n])$ is at most 1 for all $1 \le x \le n$ in this case,
$Y(x, P)$ returns the value $y$ which satisfies $(x, y) \in P$ if such $(x, y)$ exists in $P$.
Also, $Y(x, P)$ query can be answered in $O(\log n / \log {\log {n}})$ time using the data structure of Lemma~\ref{2drange}. Similarly, if no two input points have the same $y$-coordinate, then we can define the query $X(y, P) = \report{}(P, [1, n] \times [y, y])$, which can be answered in $O(\log n / \log {\log {n}})$ time.

\end{remark2}
\begin{figure}
\begin{center}
\includegraphics[scale=0.45]{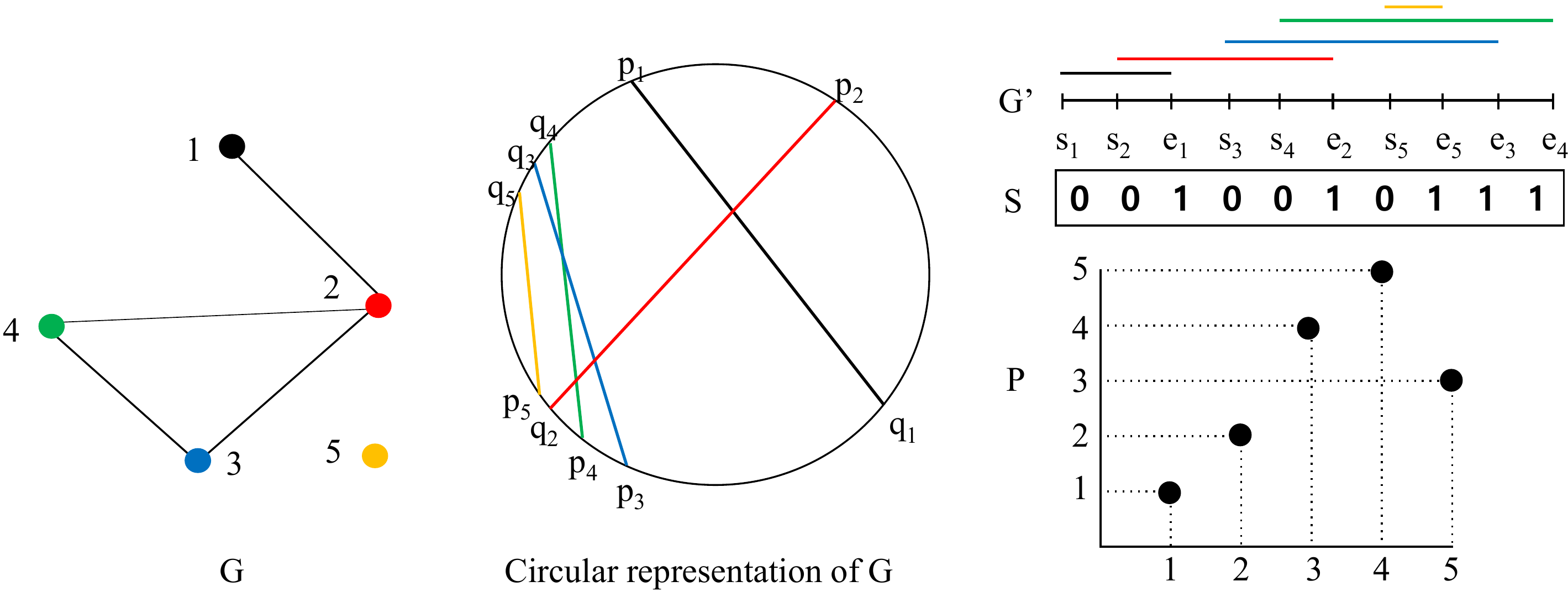}
\end{center}
\caption{A circle graph $G$ with $n=5$ vertices (left), circular representation of $G$ (middle), its corresponding overlap graph $G'$ (right top), and $S$ and $P$ for the representation of $G$.}
\label{figure:circle_upper}
\end{figure}
Now we describe the alternative succinct representation for circle graphs. Suppose $G = (V, E)$ is an intersection graph of the set of chords $\mathcal{C} = \{c_1 = (p_1, q_1), c_2 = (p_2, q_2) \dots, c_n = (p_n, q_n)\}$ of the circle $C$ 
where all the points in the set $P = \{p_1, p_2, \dots p_n\} \cup \{q_1, q_2, \dots, q_n\}$ on $C$ are distinct. Note that when $G$ is given, we can find the corresponding $\mathcal{C}$ in $O(n^2)$ time~\cite{DBLP:journals/jal/Spinrad94}.
Next, consider a bijective map $f$ from $P$ to $\{1, 2, \dots, 2n\}$ where for any $v \in P$, $f(v) = v'$ if and only if 
$v$ is the $v'$-th point from $p_1$ according to the clock-wise direction (we define $f(p_1) = 1$). 
Then for $1 \le i \le n$, $f$ maps the chord $c_i$ to the interval $I_i = [s_i, e_i] \subset [1, 2n]$, where $s_i = \min{}(f(p_i), f(q_i))$, and $e_i = \max{}(f(p_i), f(q_i))$.  
Note that since all the points in $P$ are disjoint, $\{f(p) \mid p \in P\} = \{1, 2, \dots, 2n\}$.
Now using the set of intervals $\mathcal{I} = \{I_1, I_2, \dots, I_n\}$, we define an \textit{overlap graph} $G' = (V', E')$ as follows: 
\begin{itemize}
    \item $V' = \{1, 2, \dots, n\}$, where for $1 \le i \le n$, the vertex $i$ corresponds to the interval $I_i$.
    \item For any vertices $i, j \in G'$, $(i, j) \in E'$ if and only if $I_i$ and $I_j$ are \textit{overlap}, i.e., $I_i \cap I_j \neq \emptyset$, $I_i \not\subset I_j$, and $I_j \not\subset I_i$.
\end{itemize}
It is well-known that $G$ and $G'$ are equal graphs (in general, any graph is a circle graph if and only if it is an overlap graph)~\cite{DBLP:journals/networks/Gavril73}. In the rest of this section, we refer to $G$ as the overlap graph $G'$. Now in what follows we describe our data structure for representing $G$.

\begin{enumerate}
    \item Let $S[1 \dots 2n]$ be a bit array of size $2n$ where for $1 \le i \le 2n$, $S[i] = 0$ (resp. $S[i] = 1$) if $i \in \{s_1, s_2, \dots, s_n\}$ (resp. $i \in \{e_1, e_2, \dots, e_n\}$). We maintain the data structure of Lemma~\ref{rankselect2} on $S$ using $2n +o(n)$ bits to support $\rank{}$ and $\select{}$ queries in $O(1)$ time. 
    \item For $1 \le i \le n$, let $e'_i = \rank{}_1(S, e_i)$. Since $\{e'_1, e'_2 \dots, e'_n\} =  \{1, 2, \dots, n\}$, we can consider the set of $n$ points $P = \{(1, e'_1), (2, e'_2), \dots (n, e'_n)\}$ on the $n \times n$ grid. We maintain $n \log n + o(n \log n)$-bit data structure of Lemma~\ref{2drange}, to answer $\cnt(P, R)$, and $Y(x,P)$ queries in $O(\log n / \log {\log n})$ time and $\report{}(P, R)$ queries in $O(k\log n / \log {\log n})$ time for any $1 \le x \le n$ and rectangular range $R$, where $k$ is the number of points in $P$ on $R$.
\end{enumerate}

See Figure~\ref{figure:circle_upper} for an example.
The total space of the above substructures takes $n \log n + \order(n \log n)$ bits, and for vertex $v \in V$, we can compute the corresponding interval $I_v = [s_v, e_v]$ in $\Order(\log n / \log {\log n})$ time by $s_v = \select{}_0 (v, S)$, and $e_v = \select{}_1(Y(v, P), S)$. 
Now for any two vertices $u, v \in G$, we show how to support $\degree{}(v)$ and $\adjacent{}(u, v)$ query in $\Order(\log n / \log {\log n})$ time and $\neighbor{}$ query in $\Order(|\degree{}(v)|\log n / \log {\log n})$ time using our representation.
\\\\
\noindent\textbf{$\degree{}\boldsymbol{(v)}$ query: } To answer $\degree{}(v)$ query, we first compute the corresponding interval $I_v$ in $\Order(\log n / \log {\log n})$ time, and count 
the number of intervals overlap with $I_v$, which is the sum of (i) the number of intervals $I_p$ where $s_p < s_v$ and $s_v < e_p < e_v$, and (ii) the number of intervals $I_p$ where $s_v < s_p < e_v$ and $e_p > e_v$.
By the definition of the set $P$, the number of intervals in (i) is same as the answer of 
$\cnt{}(P, R_1)$ which can be computed in $\Order(\log n / \log \log n )$ time, where $R_1 = [1, \rank{}_0(s_v, S)-1] \times [\rank{}_1(s_v, S)+1, \rank{}_1(e_v, S)]$. Similarly we can count the number of intervals in (ii) by returning $\cnt{}(P, R_2)$ in $O(\log n / \log \log n )$ time, where $R_2 = [v+1, \rank{}_0(e_v, S)] \times [\rank{}_1(e_v, S)+1, n]$ (note that $R_1$ and $R_2$ can be computed in $O(1)$ time when $I_v = [s_v, e_v]$ is given). Thus by Lemma~\ref{2drange}, we can answer $\degree{}(v)$ query in $O(\log n / \log {\log n})$ time in total.
\\
\noindent\textbf{$\adjacent{}\boldsymbol{(u, v)}$ query: }To answer $\adjacent{}(u, v)$ query, it is enough to check whether the corresponding intervals $I_u$ and $I_v$ are overlap or not. Since we can compute $I_u$ and $I_v$ in $\Order(\log n / \log {\log n})$ time, $\adjacent{}(u, v)$ query can be answered in $\Order(\log n / \log {\log n})$ time. 
\\
\noindent\textbf{$\neighbor{}\boldsymbol{(v)}$ query: } To answer $\neighbor{}(v)$ query, we simply report all the intervals in (i) and (ii) which are mentioned in the $\degree{}(v)$ query. Thus, we can answer $\neighbor{}(v)$ query in $\Order(|\degree{}(v)|\log n / \log {\log n})$ time by returning the first coordinates of the answer of $\report{}(P, R_1)$ and $\report{}(P, R_2)$ queries, where $R_1$ and $R_2$ are rectangular ranges in the grid which are defined same as the above.

We summarize our result in the following theorem
\begin{theorem}\label{thm:ub_circle}
Let $G$ be an unlabeled circle graph with $n$ vertices. Then there exists an $(n\log n + \order(n \log n))$-bit data structure representing $G$ that supports $\degree{}(v)$ and $\adjacent{}(u, v)$ queries in $\Order(\log n / \log {\log n})$ time, and $\neighbor{}(v)$ queries in $\Order(|\degree{}(v)| \cdot$ $\log n / \log {\log n})$ time.
\end{theorem}
\section{Trapezoid graphs}\label{sec:trapezoid}
In this section, we give the lower bound on space for representing trapezoid graphs, 
which implies that the representation of Theorem~\ref{thm:generaluppernew} gives a succinct representation of trapezoid graphs. Also we give an alternative succinct representation of trapezoid graphs, which uses the similar idea as Theorem~\ref{thm:ub_circle} to answer $\degree{}(v)$ queries independent of $|\degree{}(v)|$ (and again, it takes more time for the other two queries compared to the representation of Theorem~\ref{thm:generaluppernew}).

\subsection{Lower bound}\label{sec:trapezoidlb}
We can obtain a lower bound on the number of trapezoid graphs, 
intersection graphs of trapezoids where their corners are on two parallel lines. 

\begin{theorem}\label{lem:trapezoid_lb}
Consider a family of intersection graphs made from
$n$ trapezoids on two parallel lines.
Let $P_{n}$ denotes the total number of such graphs.
Then the following holds:
$$
\log P_{n} 
\geq 3n \log{n} -4 n \log\log n - \Order(n).
$$
\end{theorem}

\begin{figure}[bt]
\begin{center}
\includegraphics[scale=0.4]{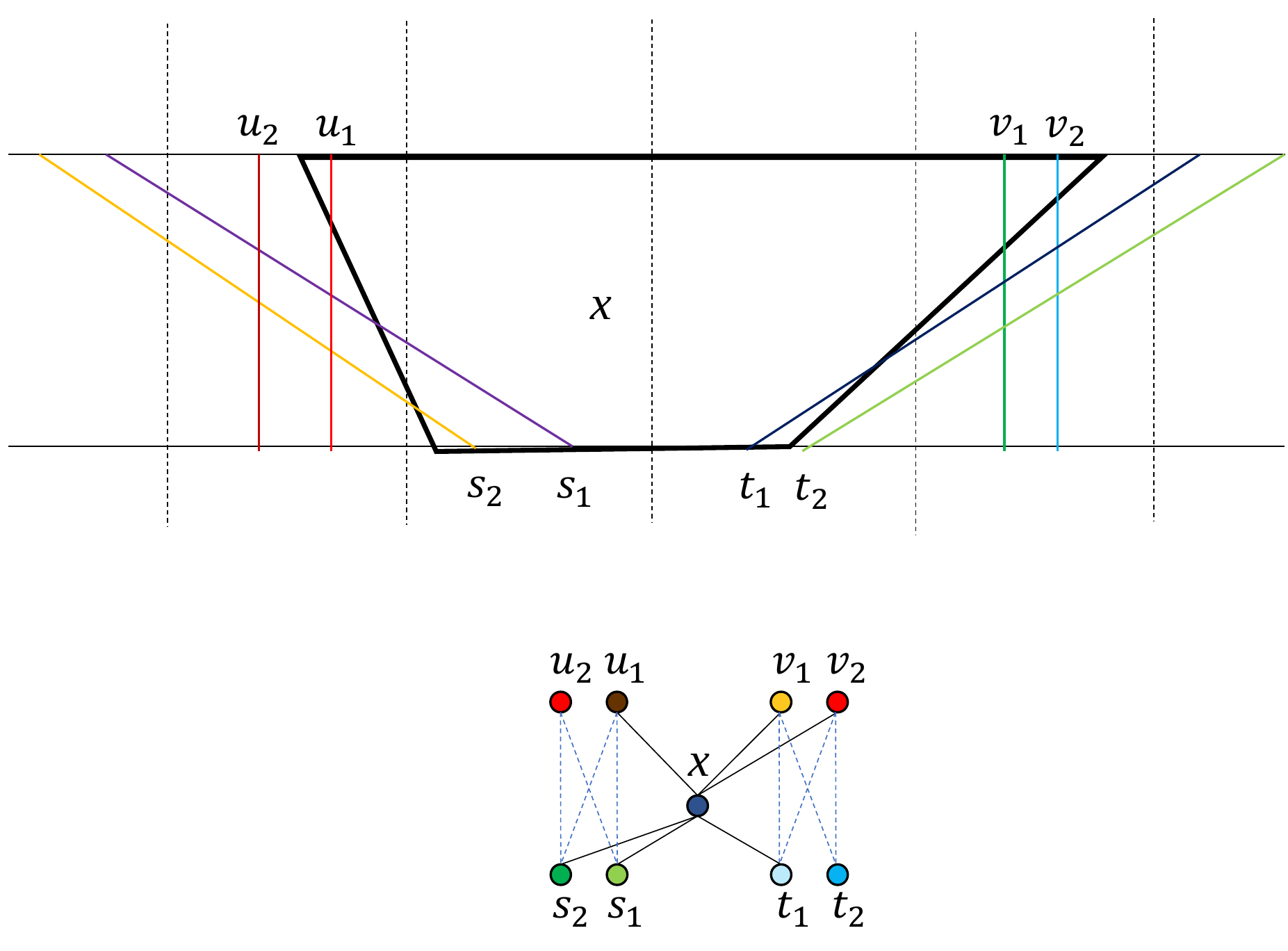}
\end{center}
\caption{A trapezoid graph (top) and its corresponding intersection graph (bottom).}
\label{fig:circle_trapezoid_lb}
\end{figure}

\begin{proof}
{Let $m = \frac{n}{\log n}$. We first consider partially-colored trapezoids with $4m$ colored trapezoids and $n-4m$ uncolored ones. We represent each colored trapezoid with a line, which could be thought as a thin trapezoid. To construct our diagrams, we divide the upper and lower lines into six pieces, $U_1,\dots, U_6$ and $L_1,\dots,L_6$, respectively, from left to right. From each of the pairs $(U_1,L_3)$, $(U_2,L_2)$, $(U_6,L_4)$, and $(U_5,L_5)$, we get $m$ colored parallel lines. Then, we draw $n-4m$ uncolored trapezoids, each of which has exactly one corner from each of the segments $U_2$, $L_3$, $L_4$, and $U_5$. We also make sure that two uncolored trapezoids do not intersect exactly the same set of colored lines. (See Figure~\ref{fig:circle_trapezoid_lb}.) Note that, in the graph corresponding to a diagram, the colored neighbors of an uncolored vertex gives us where the corners of the trapezoid corresponding to that vertex are located. There are $(m+1)^4$ possible uncolored trapezoids, and hence there are ${{(m+1)^4} \choose {n-4m}}$ different colored intersection graphs coming from this construction.  The claim follows from similar arguments used in previous lemmas.
}
\end{proof}

\subsection{Alternative succinct representation}
Our representation of trapezoid graphs uses the similar idea as the representation of 
Section~\ref{sec:ubcircle}
which uses orthogonal range queries. Suppose $G = (V, E)$ is given as the representation of $n$ trapezoids $T_1, T_2, \dots, T_n$ 
on two lines $L_1$ and $L_2$ in the two-dimensional space
which are parallel to the $x$-axis
as follows. 
For $1 \le i \le n$, $T_i$ is the trapezoid corresponding to the vertex $i \in G$, which has two points $a_i$ and $b_i$ from $L_1$ and other two points $c_i$ and $d_i$ from $L_2$ where $a_i$ is the $i$-th leftmost point on $L_1$ among $\{a_1, a_2, \dots, a_n\}$.
We interchangeably use names of points and their $x$-coordinates.
Also without loss of generality, we assume that $a_i < b_i$, $c_i < d_i$, and there is no point $p \in \{a_1, a_2, \dots, a_n\} \cup \{b_1, b_2, \dots, b_n\}$ and $q \in \{c_1, c_2, \dots, c_n\} \cup \{d_1,d_2, \dots, d_n\}$ such that the line $\overline{pq}$ is not perpendicular to neither $L_1$ nor $L_2$. 
Note that one can obtain such representation from $G$ in $O(n(n+m))$ time, where $m$ is the number of edges in $G$~\cite{DBLP:journals/dam/MertziosC11}.
We denote four sets $V_0$, $V_1$, $V_2$, and $V_3$ as $\{a_1, a_2, \dots, a_n\}$, $\{b_1, b_2, \dots, b_n\}$, $\{c_1, c_2, \dots, c_n\}$, and $\{d_1,d_2, \dots, d_n\}$ respectively. Now the following shows our representation of $G$.
\begin{enumerate}
    \item Consider an imaginary line $L_3$ parallel to $L_1$ and $L_2$.
    We project all the points in $V_0 \cup V_1 \cup V_2 \cup V_3$ orthogonally onto $L_3$, and denote the set of theses $4n$ points as $P'$. Now we consider an array $S[1, 2, \dots, 4n]$ of size $4n$ over an alphabet $\{0,1,2,3\}$ where for $1 \le i \le 4n$, $S[i] = j$ if $i$-th leftmost point in $P'$ is the point from $V_j$. We maintain the data structure of Lemma~\ref{rankselect} on $S$ using $4n +o(n)$ bits to support $\rank{}$ and $\select{}$ queries in $O(1)$ time.
    
    \item for $1 \le i \le n$, let $b'_i = \rank{}_1(S, b_i)$, $c'_i = \rank{}_2(S, c_i)$, and $d'_i = \rank{}_3(S, d_i)$.
     Then we can consider three sets of $n$ distinct points $P_1 = \{(1, b'_1), (2, b'_2), \dots (n, b'_n)\}$, $P_2 = \{(1, c'_1), (2, c'_2), \dots (n, c'_n)\}$, and $P_3 = \{(b'_1, d'_1), (b'_2, d'_2), \dots (b'_n, d'_n)\}$ on the $n \times n$ grid. We maintain $3n \log n + o(n \log n)$-bit data structure of Lemma~\ref{2drange}, to answer $\cnt{}$, and $Y$, and $\report{}$ queries on $P_1$, $P_2$, and $P_3$ efficiently. 
\end{enumerate}

The total space of the above substructures takes $3n \log n + \order(n \log n)$ bits, which is succinct by Theorem~\ref{lem:trapezoid_lb}. Also for vertex $v \in V$, we can compute the corresponding points $(v, b'_v) \in P_1$, $(v, c'_v) \in P_2$, and $(b'_v, d'_v) \in P_3$ in $\Order(\log n / \log {\log n})$ time by $b'_v = Y(v, P_1)$, , $c'_v = Y (v, P_2)$, and $d'_v = Y(b_2, P_3)$. 
Now for any two vertices $u, v \in G$, we show how to support $\degree{}(v)$ and $\adjacent{}(u, v)$ query in $\Order(\log n / \log {\log n})$ time and $\neighbor{}$ query in $\Order(|\degree{}(v)|\log n / \log {\log n})$ time using our representation.
\\\\
\noindent\textbf{$\adjacent{}\boldsymbol{(u, v)}$ query: }
To answer $\adjacent{}(u, v)$ query, it is enough to check whether the corresponding Trapezoids $T_u$ and $T_v$ are intersect or not, and it is clear that $T_u$ and $T_v$ are not intersect if and only if 
(i) $\max{}(b'_u, d'_u) < \min{}(a'_v, c'_v)$ , or (ii) $\min{}(a'_u, c'_u) > \max{}(b'_v, d'_v)$. 
Since we can compute all these eight values in $\Order(\log n / \log {\log n})$ time, $\adjacent{}(u, v)$ query can be answered in the same time. 
\\
\noindent\textbf{$\degree{}\boldsymbol{(v)}$ query: } 
To answer $\degree{}(v)$ query, we count the number of vertices in $G$ which is not adjacent to $v$, which is a total size of disjoint union of two sets of trapezoids satisfying the above conditions (i) and (ii) respectively. 
To count the number of trapezoids satisfying the condition (i), let $p_1 = (p^1_{x}, p^1_{y})$ where $p^1_{x} = (\rank{}_0(S, \max{}(b'_v, d'_v))$ and $p^1_y = \rank{}_2(S, \max{}(b'_v, d'_v))$. 
Then we can count the number of such trapezoids in $O(\log n / \log {\log n})$ time
by $\cnt{}(P_2, R_1)$ where $R_1 = [p^1_x, n] \times [p^1_y, n]$.
Similarly, let $p_2 = (p^2_{x}, p^2_{y})$ where $p^1_{x} = (\rank{}_1(S, \min(a'_v, c'_v))$ and $p^2_y = \rank{}_3(S, \min{}(a'_v, c'_v))$. 
Then we can count the number of trapezoids satisfying the condition (ii) in $O(\log n / \log {\log n})$ by $\cnt{}(P_3, R_2)$ where $R_2 = [1, p^2_x] \times [1,  p^2_y]$.
\\
\noindent\textbf{$\neighbor{}\boldsymbol{(v)}$ query: } To answer $\neighbor{}(v)$ query, we simply report all the trapezoids not in (i) and (ii) which are mentioned in the $\degree{}(v)$ query. Thus, we can answer $\neighbor{}(v)$ query in $\Order(|\degree{}(v)|\log n / \log {\log n})$ time by computing the first coordinates of the answer of $\report{}(P_2, R^1_1)$, $\report{}(P_2, R^2_1)$, $\report{}(P_3, R^1_2)$, and $\report{}(P_3, R^2_2)$ queries, where $R^1_1 = [1, p^1_x-1] \times [1, n]$, $R^2_1 = [ p^1_x, n] \times [1, p^1_y-1]$, $R^1_2 = [1, n] \times [p^2_y+1, n]$, and $R^2_2 = [p^2_x+1, n] \times [1, p^2_y]$.
Next, we return all the trapezoids corresponding to these coordinates by simply returning such coordinates (for the case of the queries on $P_2$), or by returning the corresponding answers of $X$ query on $P_1$ (for the case of the queries on $P_3$). Note that we return the same trapezoids at most $2$ times, but it does not affect the query time asymptotically.

We summarize or result in the following theorem.
\begin{theorem}\label{thm:trapezoid_ub}
Let $G$ be an unlabeled trapezoid graph with $n$ vertices. Then
there exists a ($3n\log n + \order(n \log n))$-bit data structure representing $G$ such that $\degree{}(v)$ and $\adjacent{}(u, v)$ query can be answered in $\Order(\log n / \log {\log n})$ time, and $\neighbor{}(v)$ query can be reported in $\Order(|\degree{}(v)|\cdot \log n / \log {\log n})$ time.
\end{theorem}
\section{Conclusion and Final Remarks}\label{conclusion}

In this article we proved a unified space lower bound for several classes of intersection graphs on a circle. Subsequently, we designed succinct navigational oracles for these classes of graphs in a uniform manner, along with efficient support for queries such as degree, adjacency and neighborhood. We conclude with the following open problem: 
can we improve the query times of our data structures, possibly to constant time?

\bibliographystyle{elsarticle-num-names}
\bibliography{ref}

\begin{thebibliography}{30}
\providecommand{\natexlab}[1]{#1}
\providecommand{\url}[1]{\texttt{#1}}
\providecommand{\urlprefix}{URL }
\expandafter\ifx\csname urlstyle\endcsname\relax
  \providecommand{\doi}[1]{doi:\discretionary{}{}{}#1}\else
  \providecommand{\doi}[1]{doi:\discretionary{}{}{}\begingroup
  \urlstyle{rm}\url{#1}\endgroup}\fi
\providecommand{\bibinfo}[2]{#2}

\bibitem[{Spinrad(2003)}]{spinrad}
\bibinfo{author}{J.~P. Spinrad}, \bibinfo{title}{Efficient graph
  representations}, vol.~\bibinfo{volume}{19} of \emph{\bibinfo{series}{Fields
  Institute monographs}}, \bibinfo{publisher}{American Mathematical Society},
  \bibinfo{year}{2003}.

\bibitem[{Haj\'os(1957)}]{Hajos}
\bibinfo{author}{G.~Haj\'os}, \bibinfo{title}{\"{U}ber eine Art von Graphen},
  \bibinfo{journal}{Int. Math. Nachr.} \bibinfo{volume}{11}
  (\bibinfo{year}{1957}) \bibinfo{pages}{1607--1620}.

\bibitem[{Golumbic(1985)}]{Golumbic85}
\bibinfo{author}{M.~C. Golumbic}, \bibinfo{title}{Interval graphs and related
  topics}, \bibinfo{journal}{Discrete Mathematics}
  \bibinfo{volume}{55}~(\bibinfo{number}{2}) (\bibinfo{year}{1985})
  \bibinfo{pages}{113--121}.

\bibitem[{Golumbic(2004)}]{Golumbic}
\bibinfo{author}{M.~C. Golumbic}, \bibinfo{title}{Algorithmic Graph Theory and
  Perfect Graphs}, \bibinfo{publisher}{Academic Press}, \bibinfo{year}{2004}.

\bibitem[{Habib et~al.(2000)Habib, McConnell, Paul, and Viennot}]{HMPV}
\bibinfo{author}{M.~Habib}, \bibinfo{author}{R.~M. McConnell},
  \bibinfo{author}{C.~Paul}, \bibinfo{author}{L.~Viennot},
  \bibinfo{title}{Lex-BFS and partition refinement, with applications to
  transitive orientation, interval graph recognition and consecutive ones
  testing}, \bibinfo{journal}{Theor. Comput. Sci.}
  \bibinfo{volume}{234}~(\bibinfo{number}{1-2}) (\bibinfo{year}{2000})
  \bibinfo{pages}{59--84}.

\bibitem[{Bouchet(1986)}]{bouchet}
\bibinfo{author}{A.~Bouchet}, \bibinfo{title}{Characterizing and recognizing
  circle graphs}, \bibinfo{howpublished}{Graph theory, Proc. 6th Yugosl.
  Semin.}, \bibinfo{year}{1986}.

\bibitem[{Even and Itai(1971)}]{even}
\bibinfo{author}{S.~Even}, \bibinfo{author}{A.~Itai}, \bibinfo{title}{Queues,
  Stacks and Graphs}, \bibinfo{journal}{Theory of Machines and Computations}
  (\bibinfo{year}{1971}) \bibinfo{pages}{71--86}.

\bibitem[{Koebe(1992)}]{koebe}
\bibinfo{author}{M.~Koebe}, \bibinfo{title}{On a new class of intersection
  graphs}, \bibinfo{journal}{Ann. Discrete Mathematics}  (\bibinfo{year}{1992})
  \bibinfo{pages}{141--143}.

\bibitem[{Enright and Kitaev(2019)}]{EnrightK19}
\bibinfo{author}{J.~Enright}, \bibinfo{author}{S.~Kitaev},
  \bibinfo{title}{Polygon-circle and word-representable graphs},
  \bibinfo{journal}{Electronic Notes in Discrete Mathematics}
  \bibinfo{volume}{71} (\bibinfo{year}{2019}) \bibinfo{pages}{3--8}.

\bibitem[{Felsner et~al.(1997)Felsner, M{\"{u}}ller, and
  Wernisch}]{FelsnerMW97}
\bibinfo{author}{S.~Felsner}, \bibinfo{author}{R.~M{\"{u}}ller},
  \bibinfo{author}{L.~Wernisch}, \bibinfo{title}{Trapezoid Graphs and
  Generalizations, Geometry and Algorithms}, \bibinfo{journal}{Discrete Applied
  Mathematics} \bibinfo{volume}{74}~(\bibinfo{number}{1})
  (\bibinfo{year}{1997}) \bibinfo{pages}{13--32}.

\bibitem[{Navarro(2016)}]{gonzalo}
\bibinfo{author}{G.~Navarro}, \bibinfo{title}{Compact Data Structures - {A}
  Practical Approach}, \bibinfo{publisher}{Cambridge University Press},
  \bibinfo{year}{2016}.

\bibitem[{Farzan and Munro(2013)}]{FarzanM13}
\bibinfo{author}{A.~Farzan}, \bibinfo{author}{J.~I. Munro},
  \bibinfo{title}{Succinct encoding of arbitrary graphs},
  \bibinfo{journal}{Theor. Comput. Sci.} \bibinfo{volume}{513}
  (\bibinfo{year}{2013}) \bibinfo{pages}{38--52}.

\bibitem[{Munro and Raman(2001)}]{MunroR01}
\bibinfo{author}{J.~I. Munro}, \bibinfo{author}{V.~Raman},
  \bibinfo{title}{Succinct Representation of Balanced Parentheses and Static
  Trees}, \bibinfo{journal}{{SIAM} J. Comput.}
  \bibinfo{volume}{31}~(\bibinfo{number}{3}) (\bibinfo{year}{2001})
  \bibinfo{pages}{762--776}.

\bibitem[{Aleardi et~al.(2008)Aleardi, Devillers, and Schaeffer}]{AleardiDS08}
\bibinfo{author}{L.~C. Aleardi}, \bibinfo{author}{O.~Devillers},
  \bibinfo{author}{G.~Schaeffer}, \bibinfo{title}{Succinct representations of
  planar maps}, \bibinfo{journal}{Theor. Comput. Sci.}
  \bibinfo{volume}{408}~(\bibinfo{number}{2-3}) (\bibinfo{year}{2008})
  \bibinfo{pages}{174--187}.

\bibitem[{Munro and Wu(2018)}]{MunroW18}
\bibinfo{author}{J.~I. Munro}, \bibinfo{author}{K.~Wu},
  \bibinfo{title}{Succinct Data Structures for Chordal Graphs}, in:
  \bibinfo{booktitle}{ISAAC}, \bibinfo{pages}{67:1--67:12},
  \bibinfo{year}{2018}.

\bibitem[{Farzan and Kamali(2014)}]{FarzanK14}
\bibinfo{author}{A.~Farzan}, \bibinfo{author}{S.~Kamali},
  \bibinfo{title}{Compact Navigation and Distance Oracles for Graphs with Small
  Treewidth}, \bibinfo{journal}{Algorithmica}
  \bibinfo{volume}{69}~(\bibinfo{number}{1}) (\bibinfo{year}{2014})
  \bibinfo{pages}{92--116}.

\bibitem[{Gavoille and Paul(2008)}]{Gavoille}
\bibinfo{author}{C.~Gavoille}, \bibinfo{author}{C.~Paul},
  \bibinfo{title}{Optimal Distance Labeling for Interval Graphs and Related
  Graph Families}, \bibinfo{journal}{{SIAM} J. Discrete Math.}
  \bibinfo{volume}{22}~(\bibinfo{number}{3}) (\bibinfo{year}{2008})
  \bibinfo{pages}{1239--1258}.

\bibitem[{Acan et~al.(2019)Acan, Chakraborty, Jo, and
  Satti}]{DBLP:conf/wads/AcanCJS19}
\bibinfo{author}{H.~Acan}, \bibinfo{author}{S.~Chakraborty},
  \bibinfo{author}{S.~Jo}, \bibinfo{author}{S.~R. Satti},
  \bibinfo{title}{Succinct Data Structures for Families of Interval Graphs},
  in: \bibinfo{booktitle}{WADS}, vol. \bibinfo{volume}{11646} of
  \emph{\bibinfo{series}{LNCS}}, \bibinfo{publisher}{Springer},
  \bibinfo{pages}{1--13}, \bibinfo{year}{2019}.

\bibitem[{Bazzaro and Gavoille(2009)}]{BazzaroG09}
\bibinfo{author}{F.~Bazzaro}, \bibinfo{author}{C.~Gavoille},
  \bibinfo{title}{Localized and compact data-structure for comparability
  graphs}, \bibinfo{journal}{Discrete Mathematics}
  \bibinfo{volume}{309}~(\bibinfo{number}{11}) (\bibinfo{year}{2009})
  \bibinfo{pages}{3465--3484}.

\bibitem[{Koh and Ree(2007)}]{KohR07}
\bibinfo{author}{Y.~Koh}, \bibinfo{author}{S.~Ree}, \bibinfo{title}{Connected
  permutation graphs}, \bibinfo{journal}{Discrete Mathematics}
  \bibinfo{volume}{307}~(\bibinfo{number}{21}) (\bibinfo{year}{2007})
  \bibinfo{pages}{2628--2635}.

\bibitem[{Pergel(2007)}]{Pergel07}
\bibinfo{author}{M.~Pergel}, \bibinfo{title}{Recognition of Polygon-Circle
  Graphs and Graphs of Interval Filaments Is {NP}-Complete}, in:
  \bibinfo{booktitle}{WG}, \bibinfo{pages}{238--247}, \bibinfo{year}{2007}.

\bibitem[{Kratochv{\'{\i}}l and Pergel(2003)}]{DBLP:conf/gd/KratochvilP03}
\bibinfo{author}{J.~Kratochv{\'{\i}}l}, \bibinfo{author}{M.~Pergel},
  \bibinfo{title}{Two Results on Intersection Graphs of Polygons}, in:
  \bibinfo{booktitle}{GD}, \bibinfo{pages}{59--70}, \bibinfo{year}{2003}.

\bibitem[{McKee and McMorris(1999)}]{terry}
\bibinfo{author}{T.~A. McKee}, \bibinfo{author}{F.~R. McMorris},
  \bibinfo{title}{Topics in Intersection Graph Theory},
  \bibinfo{publisher}{SIAM Monographs on Discrete Mathematics and
  Applications}, \bibinfo{year}{1999}.

\bibitem[{Clark and Munro(1996)}]{Clark:1996:EST:313852.314087}
\bibinfo{author}{D.~R. Clark}, \bibinfo{author}{J.~I. Munro},
  \bibinfo{title}{Efficient Suffix Trees on Secondary Storage}, SODA '96,
  \bibinfo{pages}{383--391}, \bibinfo{year}{1996}.

\bibitem[{Barbay et~al.(2014)Barbay, Claude, Gagie, Navarro, and
  Nekrich}]{BCGNNalgor13}
\bibinfo{author}{J.~Barbay}, \bibinfo{author}{F.~Claude},
  \bibinfo{author}{T.~Gagie}, \bibinfo{author}{G.~Navarro},
  \bibinfo{author}{Y.~Nekrich}, \bibinfo{title}{Efficient Fully-Compressed
  Sequence Representations}, \bibinfo{journal}{Algorithmica}
  \bibinfo{volume}{69}~(\bibinfo{number}{1}) (\bibinfo{year}{2014})
  \bibinfo{pages}{232--268}.

\bibitem[{Fischer and Heun(2011)}]{DBLP:journals/siamcomp/FischerH11}
\bibinfo{author}{J.~Fischer}, \bibinfo{author}{V.~Heun},
  \bibinfo{title}{Space-Efficient Preprocessing Schemes for Range Minimum
  Queries on Static Arrays}, \bibinfo{journal}{{SIAM} J. Comput.}
  \bibinfo{volume}{40}~(\bibinfo{number}{2}) (\bibinfo{year}{2011})
  \bibinfo{pages}{465--492}.

\bibitem[{Bose et~al.(2009)Bose, He, Maheshwari, and
  Morin}]{DBLP:conf/wads/BoseHMM09}
\bibinfo{author}{P.~Bose}, \bibinfo{author}{M.~He},
  \bibinfo{author}{A.~Maheshwari}, \bibinfo{author}{P.~Morin},
  \bibinfo{title}{Succinct Orthogonal Range Search Structures on a Grid with
  Applications to Text Indexing}, in: \bibinfo{booktitle}{WADS},
  \bibinfo{pages}{98--109}, \bibinfo{year}{2009}.

\bibitem[{Spinrad(1994)}]{DBLP:journals/jal/Spinrad94}
\bibinfo{author}{J.~P. Spinrad}, \bibinfo{title}{Recognition of Circle Graphs},
  \bibinfo{journal}{J. Algorithms} \bibinfo{volume}{16}~(\bibinfo{number}{2})
  (\bibinfo{year}{1994}) \bibinfo{pages}{264--282}.

\bibitem[{Gavril(1973)}]{DBLP:journals/networks/Gavril73}
\bibinfo{author}{F.~Gavril}, \bibinfo{title}{Algorithms for a maximum clique
  and a maximum independent set of a circle graph}, \bibinfo{journal}{Networks}
  \bibinfo{volume}{3}~(\bibinfo{number}{3}) (\bibinfo{year}{1973})
  \bibinfo{pages}{261--273}.

\bibitem[{Mertzios and Corneil(2011)}]{DBLP:journals/dam/MertziosC11}
\bibinfo{author}{G.~B. Mertzios}, \bibinfo{author}{D.~G. Corneil},
  \bibinfo{title}{Vertex splitting and the recognition of trapezoid graphs},
  \bibinfo{journal}{Discrete Applied Mathematics}
  \bibinfo{volume}{159}~(\bibinfo{number}{11}) (\bibinfo{year}{2011})
  \bibinfo{pages}{1131--1147}.

\end{thebibliography}
\end{document}